\documentclass[AMA,Times1COL]{WileyNJDv5} 

\articletype{Research Article}%

\received{Date Month Year}
\revised{Date Month Year}
\accepted{Date Month Year}
\journal{Journal}
\volume{00}
\copyyear{2023}
\startpage{1}

\usepackage{amsmath,amssymb,amsfonts}
\usepackage{mathtools}
\usepackage{subcaption} 
\usepackage{accents} 
\usepackage[commandnameprefix=ifneeded]{changes}

\NewDocumentCommand{\evalat}{sO{\big}mm}{%
  \IfBooleanTF{#1}
   {\mleft. #3 \mright|_{#4}}
   {#3#2|_{#4}}%
}

\newcommand{\partevalat}[3]{\ensuremath{ \evalat[\bigg]{\frac{\partial #1}{\partial #2}}{#3} }} 

\newcommand{\rvect}[1]{\begin{bmatrix} #1 \end{bmatrix}}

\newcommand{\deldel}[2]{\frac{\partial #1}{\partial #2}}

\newcommand{\undbar}[1]{\ensuremath{\underaccent{\bar}{#1}}}

\newcommand{\zerobd}{\ensuremath{\bold{0}}}
\newcommand{\Dbd}{\ensuremath{\bold{D}}}

\newcommand{\Ubd}{\ensuremath{\bold{U}}}
\newcommand{\Xbd}{\ensuremath{\bold{X}}}
\newcommand{\Zbd}{\ensuremath{\bold{Z}}}

\newcommand{\Ac}{\ensuremath{\mathcal{A}}}
\newcommand{\Bc}{\ensuremath{\mathcal{B}}}

\newcommand{\Ic}{\ensuremath{\mathcal{I}}}
\newcommand{\Yc}{\ensuremath{\mathcal{Y}}}

\newcommand{\Dc}{\ensuremath{\mathcal{D}}}

\newcommand{\Lc}{\ensuremath{\mathcal{L}}}

\newcommand{\Zc}{\ensuremath{\mathcal{Z}}}
\newcommand{\bigE}{\ensuremath{\mathcal{E}}}

\newcommand{\ch}{\ensuremath{\hat{c}}}

\newcommand{\xh}{\ensuremath{\hat{x}}}
\newcommand{\uh}{\ensuremath{\hat{u}}}

\newcommand{\yh}{\ensuremath{\hat{y}}}
\newcommand{\Vh}{\ensuremath{\hat{V}}}
\newcommand{\Gh}{\ensuremath{\hat{G}}}

\newcommand{\Yhc}{\ensuremath{\hat{\mathcal{Y}}}}

\newcommand{\Dd}{\ensuremath{\Delta d}}

\newcommand{\Dz}{\ensuremath{\Delta z}}

\newcommand{\Ghi}{\Gh^i}

\newcommand{\Vhi}{\Vh^i}

\newcommand{\Real}{\ensuremath{\mathbb{R}}}

\newcommand{\nx}{\ensuremath{n_x}}
\newcommand{\nd}{\ensuremath{n_d}}

\newcommand{\Dbold}{\ensuremath{\mathbf{D}}}

\newcommand{\Ubold}{\ensuremath{\mathbf{U}}}
\newcommand{\Xbold}{\ensuremath{\mathbf{X}}}

\newcommand{\Dseq}{\Dbold}
\newcommand{\Useq}{\Ubold}
\newcommand{\Xseq}{\Xbold}

\newcommand{\xref}{\ensuremath{x^r}}
\newcommand{\uref}{\ensuremath{u^r}}
\newcommand{\dref}{\ensuremath{d^r}}
\newcommand{\yref}{\ensuremath{y^r}}
\newcommand{\zref}{\ensuremath{z^r}}

\newcommand{\OrignalMPCProblem}{\textbf{\mbox{Problem~1}}}
\newcommand{\RNMPCProblem}{\textbf{\mbox{Problem~2}}}
\newcommand{\AlgGenericRNMPC}{\mbox{Algorithm \ref{alg:Generic_RNMPC}}}
\newcommand{\AlgOptimValidRefTrajectory}{\mbox{Algorithm \ref{alg:OptimValidRefTrajectory}}}
\newcommand{\AlgErrorSetsComp}{\mbox{Algorithm \ref{alg:ErrorSetsComputation}}}

\newcommand{\AlgLagRemOverapprox}{\mbox{Algorithm \ref{alg:LagRemOverapprox}}}

\newcommand{\SLRNMPC}{\mbox{SL-MPC}}
\newcommand{\IPOPTOne}{\mbox{IPOPT-MPC1}}
\newcommand{\IPOPTTwo}{\mbox{IPOPT-MPC2}}
\newcommand{\SLS}{\mbox{SLS}}

\newcommand{\MainTestCase}{Test Case 1}

\DeclareMathOperator{\IH}{IH} 
\DeclareMathOperator{\diag}{diag} 

\newcommand{\Uref}{\ensuremath{\Ubd^r}}
\newcommand{\Xref}{\ensuremath{\Xbd^r}}
\newcommand{\Dref}{\ensuremath{\Dbd^r}}

\algnewcommand{\Initialize}[1]{%
  \State \textbf{Initialize:}
  \Statex \hspace*{\algorithmicindent}\parbox[t]{.8\linewidth}{\raggedright #1}
}

\algnewcommand\True{\textbf{true}\space}
\algnewcommand\False{\textbf{false}\space}

\algnewcommand\Success{\texttt{success}\space}
\algnewcommand\maxIters{\texttt{maxIters}\space}

\newcommand{\mytitle}{Robust Model Predictive Control for nonlinear discrete-time systems using iterative time-varying constraint tightening
}

\DeclareRobustCommand{\VAN}[3]{#2} 

\begin{document}
\title{\mytitle{}}

\author[1]{Daniel D. Leister}
\author[2]{Justin P. Koeln}

\authormark{LEISTER \textsc{et al.}}
\titlemark{\mytitle{}}

\address[1]{\orgdiv{Mechanical Engineering Department}, \orgname{The University of Texas at Dallas}, \orgaddress{\state{Texas}, \country{USA}}}
\address[2]{\orgdiv{Mechanical Engineering Department}, \orgname{The University of Texas at Dallas}, \orgaddress{\state{Texas}, \country{USA}}}

\corres{Corresponding author Justin Koeln. \email{justin.koeln@utdallas.edu}}

\fundingInfo{Office of Naval Research. Award number:N00014-22-1-2247.}

\abstract[Abstract] {Robust Model Predictive Control (MPC) for nonlinear systems is a problem that poses significant challenges as highlighted by the diversity of approaches proposed in the last decades. Often compromises with respect to computational load, conservatism, generality, or implementation complexity have to be made, and finding an approach that provides the right balance is still a challenge to the research community. This work provides a contribution by proposing a novel shrinking-horizon robust MPC formulation for nonlinear discrete-time systems. By explicitly accounting for how disturbances and linearization errors are propagated through the nonlinear dynamics, a constraint tightening-based formulation is obtained, with guarantees of robust constraint satisfaction. The proposed controller relies on iteratively solving a Nonlinear Program (NLP) to simultaneously optimize system operation and the required constraint tightening. Numerical experiments show the effectiveness of the proposed controller with three different choices of NLP solvers as well as significantly improved computational speed, better scalability, and generally reduced conservatism when compared to an existing technique from the literature.}


\maketitle


\section{Introduction}
\label{sect:introduction}

The ability of Model Predictive Control (MPC) to optimize system operation while adhering to state and actuator constraints is one of the main reasons for its widespread popularity. Although under certain conditions nominal MPC possesses some inherent robustness to external disturbances \cite{limonmarruedoStabilityAnalysisSystems2002}, many applications require stronger robustness guarantees. This motivated the development of robust MPC formulations, where the potential effect of external unknown disturbances is explicitly taken into account and robust constraint satisfaction can be guaranteed for all future time steps. For linear systems, many robust MPC approaches rely on modifying the optimization problem by tightening the state and input constraints such that, by ensuring that the nominal predicted state trajectory is restricted to a tighter set of constraints, the real disturbed trajectories are still within the original constraints \cite{richardsRobustModelPredictiveLTV2005} or within a ``tube'' contained in the original constraints \cite{mayneRobustModelPredictive2005}. 

Multiple approaches have been proposed to adapt robust MPC strategies for the nonlinear case, typically by applying some form of constraint tightening \cite{bravoRobustMPCConstrained2006,pinRobustModelPredictive2009,cannonRobustTubesNonlinear2011,mayneTubebasedRobustNonlinear2011,zhaoRobustContractiveEconomic2018,kohlerComputationallyEfficientRobust2021,moratoRobustNonlinearPredictive2021,messererEfficientAlgorithmTubebased2021,doff-sottaDifferenceConvexFunctions2022,leemanRobustOptimalControl2023,kimJointSynthesisTrajectory2024}, with few exceptions \cite{murilloIteratedNonlinearModel2016}. However, striking a practical balance between computational complexity, conservatism, and implementation complexity remains a challenge for nonlinear robust MPC. Part of the challenge comes from the fact that constraint tightening requires quantifying the effect of uncertainties based on system dynamics. Doing so for nonlinear systems is not trivial and often compromises have to be made. Morato et al\cite{moratoRobustNonlinearPredictive2021}, for example, propose using the framework of Linear Differential Inclusion, which leverages many of the benefits of linearity while still capturing the nonlinear nature of the system, but note that this may yield conservative disturbance propagation sets. K{\"o}hler et al\cite{kohlerComputationallyEfficientRobust2021} propose tightening the inequality constraints through the use of scalar variables that capture the ``tube sizes'', computed based on level sets of Lyapunov functions. The approach though requires that the system be incrementally stabilizable. Mayne et al\cite{mayneTubebasedRobustNonlinear2011} somewhat simplify the problem of finding tightened constraint sets by defining these sets as scaled-down versions of the original input and state constraint sets, although it is not clear how much conservatism this may introduce. Pin et al\cite{pinRobustModelPredictive2009} propose computing the constraint tightening based on Lipschitz constants of the nonlinear dynamics function, which tend to produce conservative results \cite{moratoRobustNonlinearPredictive2021,kohlerComputationallyEfficientRobust2021,doff-sottaDifferenceConvexFunctions2022}. Separate control and prediction horizons are proposed to overcome this issue by applying constraint tightening and a terminal set constraint only within the shorter control horizon , thereby reducing the propagation of errors\cite{pinRobustModelPredictive2009}. Doff-Sotta and Cannon\cite{doff-sottaDifferenceConvexFunctions2022} derive a Tube-MPC algorithm for difference-of-convex systems, a relatively general class of nonlinear systems. Their algorithm relies on the idea of iteratively solving a series of convex problems at each time step, where each problem solution generates a reference trajectory used to obtain a new linear system representation and a new linear feedback controller. However, their formulation does not consider external disturbances.

Constraint tightening-based robust Nonlinear MPC (NMPC) approaches commonly use a Linear Time-Varying (LTV) approximation of the nonlinear dynamics, obtained by linearizing about a reference trajectory, to facilitate the tightening of constraints. This naturally introduces the issue of having to bound linearization errors. Also, these formulations often use a linear feedback component in the control input, acting on the error between the nominal and the actual system trajectory. This provides the controller with the ability to limit the effect of disturbances and linearization errors on the nominal trajectories when predicting the future effect of disturbances. These approaches have the challenge of coping both with linearization errors \textit{and} the effect of external disturbances. Leeman et al\cite{leemanRobustOptimalControl2023} argue that the traditional approach of solving robust optimal control problems by optimizing a reference trajectory and designing a stabilizing feedback controller offline introduces conservatism. They propose lumping linearization, parameterization, and disturbance errors into sets that are parameterized by decision variables in the optimization problem. Similarly, the linear feedback gains are also part of the decision variables. This setup arguably removes complexity from offline design and reduces conservatism, with the potential drawback of increased complexity of the resulting optimization problem.

The controller proposed in this work shares similarities with some of the existing work in that i) LTV models obtained from reference trajectories form the basis of computing overapproximations of error sets \cite{messererEfficientAlgorithmTubebased2021,doff-sottaDifferenceConvexFunctions2022,leemanRobustOptimalControl2023,kimJointSynthesisTrajectory2024}, ii) an iterative procedure is used to optimize planned trajectories \cite{murilloIteratedNonlinearModel2016,messererEfficientAlgorithmTubebased2021,doff-sottaDifferenceConvexFunctions2022,kimJointSynthesisTrajectory2024}, iii) efficient set representations such as zonotopes and intervals are used to reduce computation times \cite{bravoRobustMPCConstrained2006,pinRobustModelPredictive2009}, and iv) a fallback control option is used when no adequate solution to the optimization problem can be found at certain time steps \cite{bravoRobustMPCConstrained2006}. However, this work differs from the existing works in some key aspects. The linear feedback controller and the corresponding constraint tightening are recomputed \textit{online and iteratively}, which is enabled by the use of efficient set representations and a practical approach to account for the effect of linearization errors and external disturbances. The constraint tightening approach used here is inspired by the method developed by Richards\cite{richardsRobustModelPredictiveLTV2005} for robust MPC for LTV systems, which differs from the traditional tube-based MPC formulation. With this strategy, the approach used in this work addresses the conservatism issue mentioned by Leeman et al\cite{leemanRobustOptimalControl2023} without introducing additional complexity to the optimization problem. Note that, in the work of Kim et al\cite{kimJointSynthesisTrajectory2024} and Messerer et al\cite{messererEfficientAlgorithmTubebased2021} the optimized trajectory, the gains of the linear feedback controller, and the constraint tightening are also computed iteratively in a similar fashion to the approach used in this work. However, both references make use of ellipsoids for the representation of error sets while this work uses zonotopes and intervals. Moreover, Kim et al\cite{kimJointSynthesisTrajectory2024} rely on the estimation of local Lipschitz constants and the control gains are obtained from the solution to a Semidefinite Program (SDP), while Messerer et al\cite{messererEfficientAlgorithmTubebased2021} use a Ricatti recursion procedure to compute the control gains. Although their main algorithms share similarities with the proposed technique, the resulting formulations are considerably different. A recent work by Leeman et al\cite{leemanFastSystemLevel2024} combines benefits of System Level Synthesis (SLS)\cite{leemanRobustOptimalControl2023} with the fast Ricatti recursion scheme\cite{messererEfficientAlgorithmTubebased2021} to derive a robust MPC for Linear Time-Varying systems with significant speed improvements compared to formulations based on off-the-shelf solvers. An extension to nonlinear systems is provided in the appendix of the referred work though no numerical examples are provided for the nonlinear case. Therefore, a comparison provided in Section \ref{sec:NumericalExamples} is restricted to the formulation from \citenum{leemanRobustOptimalControl2023}. 

Therefore, the main contribution of this work lies in proposing a novel shrinking-horizon robust NMPC formulation for discrete-time systems that relies on iterative time-varying constraint tightening to obtain robust constraint satisfaction guarantees and the ability to run in real time with potentially suboptimal performance. The shrinking horizon formulation is motivated by applications where the controller is expected to operate for a fixed period of time and steady-state equilibria do not simultaneously satisfy input and state constraints, such as aircraft missions that consist entirely of transient operation \cite{domanFuelFlowTopology2018}. Technical details for formulating an MPC controller with a shrinking horizon are provided by Koeln and Alleyne\cite{koelnTwoLevelHierarchicalMissionBased2018}.

The remainder of this paper is organized as follows. Section \ref{sec:ProblemStatement} defines the class of nonlinear dynamic systems considered and the control problem to be solved. Section \ref{sec:ComputationErrorSets} details the computation of error sets used for constraint tightening. Section \ref{sec:LQRConsTighteningAndRob} shows how constraint tightening based on a time-varying Linear Quadratic Regulator (LQR) can provide robust constraint satisfaction. Section \ref{sec:ProposedController} details the proposed controller formulation while Section \ref{sec:ConsiderationsPracticalImplementation} provides some considerations for practical implementation. Finally, Section \ref{sec:NumericalExamples} details results of using the proposed controller in simulations with a nonlinear system, including a performance comparison with one of the existing techniques in the literature.

\subsection{Notation}
The letters $i$, $k$, $l$, and $m$, when used in the subscripts, are used for time step indexing in the discrete-time dynamics of the models and controllers. The set of integers in the interval $ [k,k+i] \subset \mathbb{R}$ is denoted $\mathbb{Z}_k^{k+i} $.  In the context of MPC predictions, a double-index notation is used such as in $x_{k+i|k}$, where the first subscript $k+i$ is the time step index of the predicted variable and the second subscript $k$ is the time step index when such prediction is made. When all matrices have the same indexes, the shorthand notation $(A+BK)_{k+i|k}$ will be used to represent $(A_{k+i|k}+B_{k+i|k}K_{k+i|k})$. The variables $\nx$, $n_u$, and $\nd$ represent the number of states, inputs, and disturbances in the dynamic models. Superscript $r$ is used for variables of a \textit{reference} trajectory used to obtain LTV models. The variable $N$ refers to the final time step in the system operation. Bold capital letters such as $\Xseq$ represent a collection of vectors that refer to a trajectory in time. The variable $\Xseq_k=\{x_{k|k},...,x_{k+N|k}\}$, with $x_{k+i|k} \in \mathbb{R}^{\nx}$, refers to a state trajectory, $\Useq_k =\{u_{k|k},...,u_{k+N-1|k}\}$, with $u_{k+i|k} \in \mathbb{R}^{n_u}$, to an input trajectory, and $\Dseq_k=\{d_{k},...,d_{k+N-1}\}$, with $d_{k+i} \in \mathbb{R}^{\nd}$, to a disturbance trajectory. The operation $\Ac \oplus \Bc = \{a + b \ |  \ a \in \Ac \text{ and } b \in \Bc\}$ is the Minkowski sum and  $\Ac \ominus \Bc = \{a \ | \ (a + b) \in \Ac, \ \forall b \in \Bc\}$ is the Pontryagin difference of sets $\Ac$ and $\Bc$. The combined product of matrices $M_i$ is represented with $\prod_i^k M_i$, where matrices with increasing indexes are  \textit{multiplied to the left}. The symbol $I$ is used for the identity matrix of proper dimensions. The $\diag(x)$ operator is used to represent a diagonal matrix whose diagonal elements are given by the elements of the vector $x$.



\section{Definitions and problem statement} \label{sec:ProblemStatement}

Consider a system that operates for a fixed time period with discrete-time nonlinear dynamics given by
\begin{align}
    x_{k+1} = f(x_k,u_k,d_k), \label{eq:NonlinGenericDynamics}
\end{align}
where the nonlinear mapping $f(\cdot)$ is twice continuously differentiable, $x_k \in \Real^{\nx}$ represents the system states, $u_k \in \Real^{n_u}$ the inputs, and $d_k \in \Real^{\nd}$ the external disturbances. The system operation is considered to happen over $N \in \mathbb{Z}_+$ time steps such that $k \in \mathbb{Z}_0^{N}$.

Also, the system is subject to output constraints of the form
\begin{align}
    y_k  = (C_kx_k + D_k u_k) \in \Yc_k,  \label{eq:OrigOutputConstraints}
\end{align}
where $C_k$ and $D_k$ are suitably defined matrices such that the output constraint sets $\Yc_{k}$ represent input and state constraints of the system being controlled or any linear combination thereof. With $ k \in \mathbb{Z}_0^{N} $, it is assumed that $D_N = 0$ in \eqref{eq:OrigOutputConstraints} such that $ y_N \in \Yc_N $ imposes a terminal constraint on $ x_N $ alone.

\begin{assumption} \label{assump:YckCompactAndClosed}
    Sets $\Yc_k$ are assumed to be compact and convex sets of suitable dimensions.
\end{assumption}

At any time step $k$, the total disturbance to the system can be split into a known component and an unknown component
\begin{align}
    d_k = \dref_k + \Dd_k.
\end{align}
Here, $\dref_k$ is the known part (i.e. the reference disturbance) and refers to the expected or predicted nominal disturbance trajectory. The term $\Dd_k$ is the unkown component. 
\begin{assumption}\label{assump:knownBoundsDisturb}
    The disturbance deviation $\Dd_{k}$ is bounded, for all $k\in \mathbb{Z}_0^{N-1}$, by some known compact set $\Dc_{k}$ centered at the origin.
\end{assumption}

The goal of this work is to derive a nonlinear shrinking-horizon robust MPC to control system \eqref{eq:NonlinGenericDynamics} while providing robustness guarantees for any possible realization of disturbance trajectory $\Dseq_0$. Ultimately, the goal is to find a controller that optimizes system operation according to the optimization problem

\OrignalMPCProblem:
\begin{subequations}
\begin{align} 
    	& \min_{\Xseq_0,\Useq_0} \ell(\Xseq_0,\Useq_0), \\
    	&\text{s.t.} \, \forall k \in \mathbb{Z}_0^{N-1} \text{ and } \forall \Dd_{k} \in  \Dc_{k}, \nonumber \\
    	&\quad \quad x_{k+1} =f(x_{k},u_{k},d^r_{k}+\Dd_{k}),  \label{eq:NonlinearDynamicsConstraint} \\
        &\quad \quad y_{k} =  C_{k}x_{k} + D_{k} u_{k}, \label{eq:OutputEquationConstraint}  \\
    	&\quad \quad y_{k} \in \Yc_{k}, \; y_{N} \in \Yc_{N},\label{eq:OutputConstraint} 
\end{align}
\end{subequations}
where $\ell(\Xseq_0,\Useq_0)$ is some convex function of the state and input variables and $x_{0}$ is a known initial condition.



\section{Computation of error sets} \label{sec:ComputationErrorSets}

Since error sets play a key role in tightening constraints to achieve robustness, this section details how these sets can be computed based on bounds on linearization errors and disturbances. The error sets are obtained based on an assumed feedback behavior of the controller for future time steps as well as how the effects of linearization errors and disturbances are propagated through the linearized system dynamics. 

\subsection{Linearized system model}
Before deriving a linearized model of the system, a reference trajectory is defined as follows. 

\begin{definition}[Reference Trajectory] \label{def:ReferenceTrajectory}
In the context of this work, a reference trajectory $\Zbd^r_k=\{z^r_k,...,z^r_{N-1}\}$, where $\zref_k = (\xref_k,\uref_k,\dref_k)$, is a trajectory around which the dynamics in \eqref{eq:NonlinGenericDynamics} are linearized to obtain a LTV model of the system. Reference trajectories do not necessarily satisfy the dynamics in \eqref{eq:NonlinGenericDynamics}, unless otherwise stated.
\end{definition}

Considering some reference trajectory $\Zbd^r_k$, the nonlinear system dynamics \eqref{eq:NonlinGenericDynamics} can be rewritten using a Taylor expansion as follows
\begin{align}
    x_{k+1} = A_k x_k + B_k u_k + V_k d_k + c_k + \zeta(\xi_k), \label{eq:LTVGenericDynamics} 
\end{align}
where 
\begin{subequations} \label{eq:linearizationMatrices}
\begin{align} 
& A_k = \partevalat{f(x,u,d)}{x}{(\xref_k,\uref_k,\dref_k)}, \\
& B_k = \partevalat{f(x,u,d)}{u}{(\xref_k,\uref_k,\dref_k)}, \\
& V_k = \partevalat{f(x,u,d)}{d}{(\xref_k,\uref_k,\dref_k)}, \\
& c_k = f(\xref_k,\uref_k,\dref_k) - A_k \xref_k - B_k \uref_k - V_k \dref_k,
\end{align}
\end{subequations}
and $\zeta(\xi_k)$ is the Lagrange remainder (linearization error) of the Taylor expansion, with $\xi_k \in \Zc^{\xi}_k = \{ \zref_k + \alpha(z_k - \zref_k)| \alpha \in [0,1]  \}$.

A nominal LTV system representation is obtained when the linearization error is ignored and only reference disturbances are considered in \eqref{eq:LTVGenericDynamics} 
\begin{align}
    \xh_{k+1} = A_k\xh_k + B_k \uh_k + \ch_k,  \label{eq:NominalLTVDynamics}
\end{align}
where
\begin{align*} 
 \ch_k = f(\xref_k,\uref_k,\dref_k) - A_k \xref_k - B_k \uref_k.
\end{align*}

\subsection{Derivation of error sets from linearized dynamics}

The dynamics in \eqref{eq:NonlinGenericDynamics} are said to represent the \textit{nominal nonlinear dynamics} if the system evolution is considered when subject to a reference disturbance $\dref_k$ and not the actual disturbance $d_k$. 
The second subscript here is applied to $x_{k+i|k}$ to make clear that this is the time step when the error is being predicted.
The error between the actual states and the nominal LTV state trajectory is defined as
\begin{align}
    e_{k+i|k} = x_{k+i|k} - \xh_{k+i|k}. \label{eq:TotalErrorDefinition} 
\end{align}

In order to obtain upper bounds on this error, the controller is considered to take the form
    \begin{align}
       u_{k+i|k} &= \uh_{k+i|k} + K_{k+i|k}e_{k+i|k}, \label{eq:MPCplusLQRControl}
    \end{align}
where, similar to the derivation by Richards\cite{richardsRobustModelPredictiveLTV2005}, the linear feedback control matrices $K_{k+i|k}$ are obtained from a time-varying finite-horizon LQR controller formulation as detailed in Appendix~\ref{sec:LQRControllerEqs}. 

Now note that for the current time step, when $i=0$, $e_{k|k} =0$ due to the fact that all state trajectories start from the current measured state, which is assumed to be known. In other words: $x_{k|k} = \xh_{k|k}$. Therefore, starting from the next time step, where  $e_{k+1|k} = x_{k+1|k} - \xh_{k+1|k}$, and using \eqref{eq:LTVGenericDynamics} and \eqref{eq:MPCplusLQRControl} to propagate the error dynamics forward, it can be shown that
\begin{equation} \label{eq:ErrorDynamics}
	e_{k+i|k} =  \sum_{l=0}^{i-1} \Vhi_{k+l|k}\Dd_{k+l} + \sum_{l=0}^{i-1} \Ghi_{k+l|k}\zeta(\xi_{k+l|k}), 
\end{equation}
with
\begin{align*}
    \Gh^i_{k+l|k}  &= \begin{cases*}
		\left({\displaystyle\prod_{m=l+1}^{i-1}(A+BK)_{k+m|k}} \right) ,& \text{ for $l\in \mathbb{Z}_0^{i-2}$}, \\
		I,& \text{ for $l=i-1$},
	\end{cases*}     \\
    \Vhi_{k+l|k}  &= \Gh^i_{k+l|k}  V_{k+l|k},\text{ for $l\in \mathbb{Z}_0^{i-1}$.} 
\end{align*}

Using Assumption \ref{assump:knownBoundsDisturb} and \eqref{eq:ErrorDynamics}, the total error at each time step can be bound by
\begin{align}
	\bigE_{k+i|k} = \bigoplus_{l=0}^{i-1}\left( \Vhi_{k+l|k}\Dc_{k+l} \oplus \Ghi_{k+l|k} \Lc(\Zc_{k+l|k})\right),  \label{eq:TotalErrorSetsEq}
\end{align}
where $\Zc_{k+l|k}$ is such that $\Zc^{\xi}_{k+l|k}\subseteq \Zc_{k+l|k}$ and $\Lc$ is a set operator that overapproximates the set of possible values of the Lagrange remainder $\zeta(\xi_{k+l|k})$ given $\Zc_{k+l|k}$. There are multiple methods of implementing the $\Lc$ operator to compute overapproximations for linearization errors \cite{althoffCORA2022Manual2022}. The method used in this work is detailed in Section \ref{sec:BoundingLagrangeRemainders}.

To obtain the sets $\Zc_{k+l|k}$, consider the following definition for the elements of the set
\begin{align*}
    z_{k+l|k} = \zref_{k+l|k} + \Dz_{k+l|k},
\end{align*}
with
\begin{align*}
\Dz_{k+l|k} &= z_{k+l|k} - \zref_{k+l|k} = \begin{bmatrix} x_{k+l|k} - \xref_{k+l|k} \\ u_{k+l|k} - \uref_{k+l|k} \\ d_{k+l} - \dref_{k+l|k}\end{bmatrix} . 
\end{align*}

\begin{assumption}\label{assump:OpenLoopControlComponent}
    The open-loop control component in \eqref{eq:MPCplusLQRControl} is taken as the reference input trajectory, i.e.,  $\uh_{k+i|k} = \uref_{k+i|k}$. 
\end{assumption}

If Assumption \ref{assump:OpenLoopControlComponent} holds and considering that $\xh_{k+l|k}$ is obtained from the LTV dynamics linearized around $\xref_{k+l|k}$ and both consider the nominal reference disturbance, then $\Dz_{k+l|k}$ can be rewritten as
\begin{align}
\Dz_{k+l|k} &=  \begin{bmatrix} x_{k+l|k} - \xh_{k+l|k} \\ u_{k+l|k} - \uh_{k+l|k} \\ d_{k+l} - \dref_{k+l|k}\end{bmatrix}. \label{eq:deltaz_definition}
\end{align}

Now using \eqref{eq:TotalErrorDefinition} and \eqref{eq:MPCplusLQRControl}, $\Dz_{k+l|k}$ becomes
\begin{align*}
   \Dz_{k+l|k} = \begin{bmatrix} e_{k+l|k} \\ K_{k+l|k}e_{k+l|k} \\ \Dd_{k+l} \end{bmatrix}. 
\end{align*}
Therefore, the sets $\Zc_{k+l|k}$ can be written as
\begin{align}
    \Zc_{k+l|k} = \zref_{k+l|k} \oplus \left( \bigE_{k+l|k} \times K_{k+l|k}\bigE_{k+l|k} \times \Dc_{k+l} \right). \label{eq:ExpressionSetZc}
\end{align}
From \eqref{eq:TotalErrorSetsEq} and \eqref{eq:ExpressionSetZc} , it can be seen that the sets $\bigE$ depend on the sets $\Zc$ computed for previous time steps while the sets $\Zc_{k+l|k}$ depend on the corresponding $\bigE_{k+l|k}$ sets. \AlgErrorSetsComp{} summarizes the steps needed to compute these sets sequentially based on a reference trajectory.

\begin{algorithm}[h]
	\caption{\enskip Error Sets Computation}
	\label{alg:ErrorSetsComputation}
	\begin{algorithmic}[1]
	   \Require $\zref_{k+i|k}$ for all $i\in \mathbb{Z}_0^{N-1}$
        \State $\bigE_{k|k} = \zerobd$ \Comment{due to $e_{k|k}=0$}
        \State  $\Zc_{k|k} = \zref_{k|k} \oplus  \left( \zerobd \times \zerobd \times \Dc_{k} \right)$ \Comment{from \eqref{eq:ExpressionSetZc}}
		\For{$i=1...N-1$}
			\State compute $\bigE_{k+i|k}$ using \eqref{eq:TotalErrorSetsEq} \label{line:AlgEsets_computebigE}
            \State compute $\Zc_{k+i|k}$ using \eqref{eq:ExpressionSetZc}
		\EndFor
        \State compute $\bigE_{k+N|k}$ using \eqref{eq:TotalErrorSetsEq} \label{line:AlgEsets_computebigE_N}
		\State \Return $\bigE_{k+i|k}$ for all $i\in \mathbb{Z}_1^{N}$ 
	\end{algorithmic}
\end{algorithm}



\section{LQR-based constraint tightening and robust constraint satisfaction} \label{sec:LQRConsTighteningAndRob}

With the method outlined in \AlgErrorSetsComp{} to compute error sets, this section demonstrates how the controller defined in \eqref{eq:MPCplusLQRControl} provides a suboptimal control law that robustly satisfies the system's output constraints, given a reference trajectory that satisfies certain conditions at time step $k=0$. 

\subsection{Preliminaries} \label{sec:Preliminaries}

\begin{assumption}\label{assump:Zk_Lc_containment}
    Given two sets $\Ac$ and $\Bc$, if $\Ac \subseteq \Bc$ then $\Lc(\Ac) \subseteq \Lc(\Bc)$, with $\Lc$ the set operator defined in Section \ref{sec:ComputationErrorSets}.
\end{assumption}

\begin{proposition} [Morphological Opening Property \cite{castroAnalyticalSolutionsMinkowski2013}]
Given compact and closed sets $\Ac$ and $\Bc$ containing the origin, then
\begin{align}
    (\Ac \ominus \Bc) \oplus \Bc \subseteq  \Ac. \label{eq:PropPontryag4}
\end{align}
    
\end{proposition}

Also, consider the tightened output constraint sets $\Yhc_{k+i|k}$, which are related to the original output constraints such that
\begin{align}
    \Yhc_{k+i|k} = \Yc_{k+i} \ominus \left[(C_{k+i} + D_{k+i}K_{k+i|k})\bigE_{k+i|k}\right]. \label{eq:TightenedOutputConstraintSets}
\end{align}
The tightened constraint sets $\Yhc_{k+i|k}$ depend directly on $K_{k+i|k}$ and $\bigE_{k+i|k}$. These error sets $\bigE_{k+i|k}$, as seen in Section \ref{sec:ComputationErrorSets}, also depend on the LQR gain matrices, namely on $K_{k+l|k}$ for all $l\in\mathbb{Z}_1^{i-1}$. However, these $K_{k+i|k}$ are obtained by linearization around the references $\zref_{k+i|k} $. Therefore, the tightened constraint sets $\Yhc_{k+i|k}$ are directly related to the specific references $\zref_{k+i|k}$. 

\begin{definition}[Valid Reference Trajectory] \label{def:validRefTrajectory}
    At time step $ k $, a reference trajectory $\Zbd^r_k =\{z^r_{k|k},...,z^r_{N-1|k}\}$ is said to be \textit{valid} if it satisfies the nonlinear dynamics \eqref{eq:NonlinGenericDynamics} and the resulting outputs from \eqref{eq:OrigOutputConstraints} satisfy the corresponding tightened output constraint sets $\Yhc_{k+i|k}$, $ i\in\mathbb{Z}_0^{N-k}$, obtained using this reference trajectory. 
\end{definition}

\subsection{Robust constraint satisfaction}
Now, the main result of this section can be stated as follows.
\begin{theorem} \label{theo:feasibleImpliesYConstraintSatisfaction}
    For any time step $k \in \mathbb{Z}_0^{N-1}$, if disturbances are bounded according to Assumption \ref{assump:knownBoundsDisturb}, a valid reference trajectory $\Zbd^r_k$ is known, and system \eqref{eq:NonlinGenericDynamics} is controlled by the input \eqref{eq:MPCplusLQRControl} with $\uh_{k+i|k}= \uref_{k+i|k}$ for all $i\in \mathbb{Z}_0^{N-k-1}$, then the actual system outputs $y_{k+i}$ will satisfy the original output constraints $\Yc_{k+i}$ for all $i\in\mathbb{Z}_0^{N-k}$. 
\end{theorem}
\begin{proof}
Consider the definition of the system output at time step $k+i$
    \begin{align*}
            y_{k+i} &= C_{k+i}x_{k+i} + D_{k+i}u_{k+i}.
    \end{align*}

    Substituting $x_{k+i}$ and $u_{k+i}$ using \eqref{eq:TotalErrorDefinition} and \eqref{eq:MPCplusLQRControl}, one obtains
    \begin{align*}
            y_{k+i}&= C_{k+i}(\xh_{k+i|k} + e_{k+i|k}) + D_{k+i}( \uh_{k+i|k} + K_{k+i|k}e_{k+i|k}),\\
            &= \yh_{k+i|k} + (C + DK)_{k+i|k}e_{k+i|k}.
        \end{align*}

        Notice that, as $\uh_{k+i|k}= \uref_{k+i|k}$, the nominal LTV state trajectory is equal to the state reference trajectory, i.e., $\xh_{k+i|k}= \xref_{k+i|k}$, since there are no linearization errors on the reference trajectory. So, $\yh_{k+i|k} = \yref_{k+i|k}$ and, from Definition \ref{def:validRefTrajectory}, it results that $\yh_{k+i|k} \in \Yhc_{k+i|k}$. Therefore
        \begin{align*}
                y_{k+i} \in \Yhc_{k+i|k} \oplus (C + DK)_{k+i|k}\bigE_{k+i|k}.
        \end{align*}
        Using \eqref{eq:TightenedOutputConstraintSets},
        \begin{align*}
                y_{k+i} \in \Yc_{k+i} &\ominus (C + DK)_{k+i|k}\bigE_{k+i|k} \oplus  (C + DK)_{k+i|k}\bigE_{k+i|k}. \label{eq:proof}
        \end{align*}
        Applying the opening property from  \eqref{eq:PropPontryag4} results in $y_{k+i} \in \Yc_{k+i}$ and the proof is complete.
\end{proof}

This result shows that, if a valid reference trajectory can be found at time step $k=0$, then applying the LQR-based control law given by \eqref{eq:MPCplusLQRControl} guarantees robust constraint satisfaction for all future time steps. This control strategy though is not necessarily optimal. The MPC controller proposed in Section~\ref{sec:ProposedController} aims at achieving better performance while keeping the robustness guarantees from Theorem \ref{theo:feasibleImpliesYConstraintSatisfaction}.

\begin{remark} 
    Note that, when deriving the LQR control gains $K_{k+i|k}$, no assumption is made with respect to the stabilizability of the underlying LTV system, originating from a reference trajectory satisfying the nonlinear dynamics. However, the LQR control in \eqref{eq:MPCplusLQRControl} acts to minimize deviations of the actual trajectory from the reference trajectory. Therefore, the less effective the LQR control is in minimizing these errors, the larger the error sets $\bigE_{k+i|k}$ and the smaller the tightened constraint sets $\Yhc_{k+i|k} $. In the extreme case, the error sets may be so large that some of the tightened constraint sets become empty and no valid reference trajectory exists. This ultimately depends on the nonlinear dynamics of the specific system being considered and the weighting matrices of the LQR control design.
\end{remark}



\section{Proposed Robust NMPC} \label{sec:ProposedController}

This work proposes a novel shrinking-horizon nonlinear robust MPC with a fallback control component. The fallback control is used when no valid solution to the MPC optimization problem can be found within the controller's sampling time, while still guaranteeing robust constraint satisfaction. The proposed controller \textit{iteratively} solves the optimization problem

\RNMPCProblem:
\begin{subequations}
\begin{align} 
    	& \min_{\Xseq_k,\Useq_k} \ell(\Xseq_k,\Useq_k), \\
    	&\text{s.t.} \, \forall i \in \mathbb{Z}_0^{N-k-1} \nonumber \\
    	&\quad \quad x_{k+i+1|k} =f(x_{k+i|k},u_{k+i|k},\dref_{k+i}), \label{eq:NonlinearDynamicsConstraint_MPC} \\
        &\quad \quad x_{k|k} = x_{k}, \label{eq:MPC_IC_constraint}  \\
        &\quad \quad y_{k+i|k} =  C_{k+i}x_{k+i|k} + D_{k+i} u_{k+i|k}, \label{eq:OutputEquationConstraint_MPC}  \\
    	&\quad \quad y_{k+i|k} \in \Yhc_{k+i|k}, \; y_{N|k} \in \Yhc_{N|k}. \label{eq:OutputConstraint_MPC} 
\end{align}
\end{subequations}
Since $\dref_{k+i}$ represents the known or expected disturbance values at each time step, \eqref{eq:NonlinearDynamicsConstraint_MPC} refers to the \textit{nominal} nonlinear dynamics of the system. Equation \eqref{eq:OutputConstraint_MPC} enforces constraining the nominal outputs to the tightened output constraint sets $\Yhc_{k+i|k}$ as defined in \eqref{eq:TightenedOutputConstraintSets}. This implies that \RNMPCProblem{} is tied to a specific reference trajectory since the constraint tightening is obtained from a reference trajectory as noted in Section \ref{sec:Preliminaries}. This trajectory though is not necessarily a \textit{valid} reference trajectory, as will become clear in the remainder of this section.

\subsection{Obtaining valid optimal reference trajectories}
Given that the control strategy in \eqref{eq:MPCplusLQRControl} provides a robust control strategy that is suboptimal, the controller proposed in this work uses \RNMPCProblem{} to obtain an optimized and valid reference trajectory at each time step, as shown in \AlgOptimValidRefTrajectory{}. The algorithm starts with a given initial reference trajectory and iteratively goes through the process of solving \RNMPCProblem{} to obtain a new reference trajectory, computing the LQR controller gains and tightened output constraints, and verifying whether the new reference trajectory is valid with respect to the tightened output constraints. The algorithm terminates either when a new optimized valid reference trajectory is found or the maximum number of iterations is exceeded, in which case it is considered to have failed to find such an optimized trajectory.

Also, within the scope of this work, no claims are made with respect to the convergence properties of the sequence of reference trajectories obtained in \AlgOptimValidRefTrajectory{} or its ability to return a new valid reference trajectory. The availability of a fallback control input is therefore essential to guarantee closed-loop robust constraint satisfaction of the proposed controller as detailed in Section \ref{sec:MainControllerAlgorithm}.

\begin{algorithm*}[h]
	\caption{Valid Reference Trajectory Optimization}
     \label{alg:OptimValidRefTrajectory}
	\begin{algorithmic}[1]
		\Require $u_0,x_0,\Xref_0,\Uref_0, \Dc_{k},  \Yc_k$, \maxIters \Comment{$\Uref_0$ and $\Xref_0$ are initial reference trajectories}
		\Initialize {$\Xref \gets \Xref_0$, $\Uref \gets \Uref_0$, $\Success \gets \False $}	
		\State Compute constraint tightening $\Yhc_{k+i|k}$ using \AlgErrorSetsComp{} and \eqref{eq:TightenedOutputConstraintSets}    \Comment{$\Uref$ and $\Xref$ used to obtain $K_{k+i|k}$}
		\State $j=0$
		\While{$j<\maxIters$}
		\State Obtain new $\Uref$ and $\Xref$ by solving \RNMPCProblem{} \label{line:callNLPSolver}
		\State Use  new $\Uref$ and $\Xref$ to update constraint tightening $\Yhc_{k+i|k}$ using \AlgErrorSetsComp{} and \eqref{eq:TightenedOutputConstraintSets}    
		\If{$\Uref$ and $\Xref$ satisfy $\Yhc_{k+i|k}$}
		\State $\Success \gets \True $
		\State \textbf{break}
		\EndIf
		\State $j \gets j+1$
		\EndWhile
		\State \Return \Success, $\Uref$, $\Xref$
	\end{algorithmic}
\end{algorithm*}

\subsection{Main controller algorithm} \label{sec:MainControllerAlgorithm}

\AlgGenericRNMPC{} summarizes the proposed controller formulation. Note that, since the \AlgOptimValidRefTrajectory{} may not converge to a valid solution within the maximum number of steps in Line \ref{line:callNLPSolver}, the controller needs to provide a fallback control input as shown in Line \ref{line:AlgRNMPC_FallbackControl}. On the other hand, if \AlgOptimValidRefTrajectory{} converges to a valid solution at some time step $k_v$, the obtained solution is used to update the reference trajectory $\Zbd^r$. This reference trajectory is stored, along with the corresponding LQR matrices, for use in future time steps to compute the fallback control input. At any time step, the proposed controller guarantees robust constraint satisfaction even if \AlgOptimValidRefTrajectory{} never again finds a valid solution, using the fallback option for the remaining time steps. The fallback option though provides a suboptimal control response, therefore it is desirable that the algorithm finds a valid solution in as many time steps as possible in order to obtain optimized control inputs. 


\begin{algorithm*}[h]
	\caption{Robust NMPC}
	\label{alg:Generic_RNMPC}
	\begin{algorithmic}[1]
	\Require initial reference trajectory $\Zbd^r_0$, $N$, $Q$, $R$
        \State $\Zbd^r \gets \Zbd^r_0$ \Comment{sets current reference trajectory}
        \For{$k=0...N-1$}
            \State Measure the current state of the system, $x_{k}$
            \State Run \AlgOptimValidRefTrajectory{} to obtain a new valid reference trajectory \label{line:AlgRNMPC_OptimStep} 
            \If{\AlgOptimValidRefTrajectory{} returned \Success } 
                \State Update $\Zbd^r$ from the new valid reference trajectory returned by \AlgOptimValidRefTrajectory{}
                \label{line:AlgRNMPC_ZrefUpdate}
                \State Apply the control input $u_{k|k} = \uref_{k|k}$
                \State $k_{v} \gets k$ \Comment{update time step of last valid solution}
                \State Store $\Zbd^r$ and $K_{k+i|k_v}$, for $i \in \mathbb{Z}_0^{N-k_v-1}$ \Comment{for use in fallback control input}
            \Else 
                \State Apply the fallback control input $u_{k|k} = \uref_{k|k_{v}} + K_{k|k_{v}}e_{k|k_{v}}$. \label{line:AlgRNMPC_FallbackControl}
            \EndIf
        \EndFor
	\end{algorithmic}
\end{algorithm*}


Next, the robustness guarantees of the proposed controller are presented.
\begin{assumption}
     At time step $k=0$, an initial valid reference trajectory $\Zbd^r_0= (\Xref_0,\Uref_0,\Dref_0)$ is available. 
     \label{assump:FindFeasibleSolAt_k_zero}
\end{assumption}

\begin{theorem} \label{theo:SystemWithSLRNMPCSatisfiesYk}
    If Assumptions \ref{assump:FindFeasibleSolAt_k_zero}  and \ref{assump:knownBoundsDisturb} hold true, then system \eqref{eq:NonlinGenericDynamics}, with control inputs given by \AlgGenericRNMPC{}, robustly satisfies the original output constraints \eqref{eq:OrigOutputConstraints} at all time steps within the system operation, i.e.,  for $k \in \mathbb{Z}_0^{N}$. 
\end{theorem}
\begin{proof}
    If Assumption \ref{assump:FindFeasibleSolAt_k_zero} is satisfied, then, at time step $k=0$ a feasible solution is available such that the current output $y_k$ satisfies \eqref{eq:OrigOutputConstraints}, with input $\uref_{k|k}$ being applied to the system. By \AlgGenericRNMPC{}, for each future time step $k \in \mathbb{Z}_1^{N-1}$ there are two distinct possibilities.
    \begin{enumerate}
         \item \textbf{No success in \AlgOptimValidRefTrajectory{}}: if \AlgOptimValidRefTrajectory{} does not converge to a valid solution at time step $k$ but converged at some $k_v<k$, then the control input $u_k = \uref_{k|k_v} + K_{k|k_v}e_{k|k_v}$ will be applied to the system and Theorem \ref{theo:feasibleImpliesYConstraintSatisfaction} guarantees constraint satisfaction according to \eqref{eq:OrigOutputConstraints}. The existence of a $k_v<k$ such that a valid reference trajectory is available is guaranteed by Assumption \ref{assump:FindFeasibleSolAt_k_zero}.
         \item \textbf{Success in \AlgOptimValidRefTrajectory{}}: if \AlgOptimValidRefTrajectory{} converges to a valid solution at time step $k$ then the input $\uref_{k|k}$ will be applied to the system and by Theorem \ref{theo:feasibleImpliesYConstraintSatisfaction} the current output $y_k$ satisfies \eqref{eq:OrigOutputConstraints} and there is a feasible input trajectory for all future time steps such that \eqref{eq:OrigOutputConstraints} is satisfied as well.
    \end{enumerate}
\end{proof}

\begin{remark}\label{rem:InitialRefTrajectory}
    Assuming a valid initial trajectory, as in Assumption \ref{assump:FindFeasibleSolAt_k_zero}, is standard in the robust NMPC literature. One way to find a valid initial reference trajectory for the proposed controller is to run \AlgOptimValidRefTrajectory{} offline with the initial conditions and replace the objective function in \RNMPCProblem{} with a constant, i.e., solving only a feasibility problem and not an optimization problem.
\end{remark}

\begin{remark}
    The robust constraint satisfaction guarantees, as provided in Theorems \ref{theo:feasibleImpliesYConstraintSatisfaction} and \ref{theo:SystemWithSLRNMPCSatisfiesYk}, are sufficient to ensure bounded-input, bounded-output (BIBO) stability. As noted in Koeln et al\cite{koelnVerticalHierarchicalMPC2020}, for many applications BIBO is not only sufficient but also preferred over asymptotic stability since it allows the controller to use the system dynamics to optimize system operation. 
\end{remark}



\section{Considerations for practical implementation} \label{sec:ConsiderationsPracticalImplementation}
This section provides details on some of the aspects of the practical implementation of the proposed controller that are relevant to achieve efficient computation. The measures defined here have been adopted to obtain the results shown in Section \ref{sec:NumericalExamples}.

\subsection{Set representations and overapproximations}
Depending on the set representation and specific algorithms used to compute the Minkowski sum and Pontryagin difference operations in Algorithm \ref{alg:ErrorSetsComputation} and the constraint tightening in \RNMPCProblem{}, the underlying numerical computations can be very inefficient. Practical implementations may require trading off accuracy for speed by using overapproximations of some of the computed sets. In this work, sets are represented as zonotopes and sometimes approximated to intervals. Zonotopes are centrally symmetric sets that can be efficiently represented in the so-called G-Rep \cite{mcmullenZonotopes1971} as $\Zc = \{c,G\}$, where $c \in \mathbb{R}^{n}$ is the center of the set and $G = [g_1 \cdots g_{n_g}]$ is a generator matrix whose columns $g_i \in \mathbb{R}^{n} $ are the generators of the zonotope and $n_g$ is the number of generators. Intervals provide one of the simplest set representations, which allow for very efficient computations although often providing more conservative approximations compared to other representations. An interval of dimension $n$ is a subset of $\mathbb{R}^n$ that can be specified by a lower bound vector $\undbar{x}$ and an upper bound vector $\Bar{x}$ such that \cite{althoffCORA2022Manual2022}

\begin{align*}
    \Ic = [\undbar{x},\Bar{x}]= \left\{x \mid \undbar{x}_i \le x_i \le \Bar{x}_i, \forall i \in \mathbb{Z}_1^{n} \right\}.
\end{align*}
Note that intervals can also be efficiently represented as zonotopes where the generator matrix is simply a diagonal matrix. 

The following measures are applied to improve the efficiency of the computations in this work.
\begin{itemize}
    \item \textbf{Zonotope interval hull overapproximation}: as discussed in Section \ref{sec:ComputationErrorSets}, multiple Minkowski sums are required for the computation of the error sets. Since in G-Rep Minkowski sums are executed by concatenating the generator matrices of the operands, with multiple consecutive operations the number of generators of the resulting sets rapidly grows. To keep the set complexity low, zonotopes can be overapproximated by converting them to intervals, which reduces the number of generators to the dimension of the set. This overapproximation of a zonotope is computed by essentially finding the interval hull of the zonotope. Given a zonotope $\Zc = \{c,G\}$, its interval hull is given by
    \begin{align*}
        \IH({\Zc}) = \left\{c, \diag\left(\sum_{i=1}^{n_g} |g_i|\right) \right\}.
    \end{align*}
    This conversion is applied to the following parts of the control algorithm: after computing each $\bigE_{k+i|k}$ in Lines \ref{line:AlgEsets_computebigE} and \ref{line:AlgEsets_computebigE_N} of \AlgErrorSetsComp{}, in Line \ref{line:AlgLagRemOA_IH_Zc} of \AlgLagRemOverapprox, and to the sets involved in the Pontryagin difference operation when computing the sets $\Yhc_{k+i|k}$ in \AlgOptimValidRefTrajectory{}.
    
    \item \textbf{Simplified Pontryagin difference}: although performing set operations such as Minkowski additions and linear mappings on zonotopes is very computationally efficient, the same is not true for the exact computation of the Pontryagin difference and it is often necessary to resort to overapproximations \cite{yangEfficientBackwardReachability2022}. However, when performed on intervals where the subtrahend is centered at the origin, the Pontryagin difference operations can be computed exactly and very efficiently. Given two intervals $\Ic_a = [\undbar{x}_a, \Bar{x}_a]$ and $\Ic_b= [-\Bar{x}_b, \Bar{x}_b]$, with $\Ic_b$ centered at the origin, their Pontryagin difference can be computed by
    \begin{align*}
        \Ic_a \ominus \Ic_b &= [\undbar{x}_a + \Bar{x}_b, \Bar{x}_a - \Bar{x}_b].
    \end{align*}
    This efficient way of computing the Pontryagin difference is applied when computing the sets $\Yhc_{k+i|k}$ in \AlgOptimValidRefTrajectory{}. Note that the Pontryagin difference in \AlgOptimValidRefTrajectory{} is only used to subtract sets centered at the origin. This is true due to i) disturbance sets $\Dc_k$ centered at the origin, ii) the fact that every error set $\bigE_{k+i|k}$ is overapproximated using intervals, and iii) the way the Lagrange remainders are bounded, as shown in Section \eqref{sec:BoundingLagrangeRemainders}, which produces sets centered at the origin. 
\end{itemize}

\subsection{Bounding Lagrange remainders \label{sec:BoundingLagrangeRemainders}}

The $i^{th}$ component of the Lagrange remainder of the Taylor expansion in \eqref{eq:LTVGenericDynamics} is given exactly by \cite{althoffReachabilityAnalysisNonlinear2008}
\begin{align*}
    \zeta_i(\xi_{k}) = \frac{1}{2}(z_{k} - \zref_{k})^\top J_i(\xi_{k})  (z_{k} - \zref_{k}),
\end{align*}
with
\begin{align*}
    J_i(\xi_{k}) = \deldel{^2f_i(\xi_{k})} {z^2}.
\end{align*}

There are multiple methods of bounding linearization errors due to the Lagrange remainder in the Taylor expansion of \eqref{eq:NonlinGenericDynamics}. The CORA toolbox \cite{althoffCORA2022Manual2022} provides functionality to compute $\Lc(\Zc_{k+i|k})$ using different techniques. In this work though, the approach used by Althoff et al\cite[Section V]{althoffReachabilityAnalysisNonlinear2008} is applied, which is summarized in \AlgLagRemOverapprox.

Line \ref{line:AlgLagRemOA_gamma} assumes that the center $c$ of $\Zc_{k+l|k}$ coincides with $\zref_{k+l|k}$, which, considering \eqref{eq:ExpressionSetZc}, is true, as long as the sets $\Dc_k$ have centers at the origin. In Line \ref{line:AlgLagRemOA_MaxHessian}, the $\max$ and absolute value operations are performed elementwise and simple interval arithmetic is used to find the maximum absolute values of each component, which is also provided in CORA. In Line \ref{line:AlgLagRemOA_Lci}, the $i^{th}$ component of the Lagrange remainder set is overapproximated by an interval. Finally, in Line \ref{line:AlgLagRemOA_LcComposition} all the subintervals are combined to form $\Lc(\Zc_{k+l|k})$.

\begin{algorithm}[h]
	\caption{Overapproximation of the Lagrange remainder}
	\label{alg:LagRemOverapprox}
	\begin{algorithmic}[1]
	\Require $\Zc_{k+l|k} = \{c,G\}, G = [g_1 \cdots g_{n_g}]$
        \State $\gamma = \sum_{k=1}^{n_g}|g_k|$ \label{line:AlgLagRemOA_gamma}
		\For{$i=1...n_x$}
            \State $\Ic_{\Zc} = \IH(\Zc_{k+l|k})$ \label{line:AlgLagRemOA_IH_Zc}
			\State $H_i \gets \max\left(|J_i(z)|\right), z \in \Ic_{\Zc}$  \label{line:AlgLagRemOA_MaxHessian}
            \State $M_i \gets \frac{1}{2}\gamma^\top H_i \gamma$
            \State $\Lc_i \gets [-M_i, M_i]$ \label{line:AlgLagRemOA_Lci}
		\EndFor
        \State $\Lc({\Zc_{k+l|k}}) = \Lc_1 \times \cdots \times \Lc_{n_x}$
		\State \Return $\Lc({\Zc_{k+l|k}})$  \label{line:AlgLagRemOA_LcComposition}
	\end{algorithmic}
\end{algorithm}



 \section{Numerical Experiments} \label{sec:NumericalExamples}
A series of simulation experiments have been performed to highlight the validity of the proposed approach as well as some of the trade-offs associated with different choices when implementing the proposed controller.

\subsection{Plant model}


\begin{figure}[t]
    \centering
    \includegraphics[trim={0cm 0cm 5cm 0},clip,width=9.4cm] {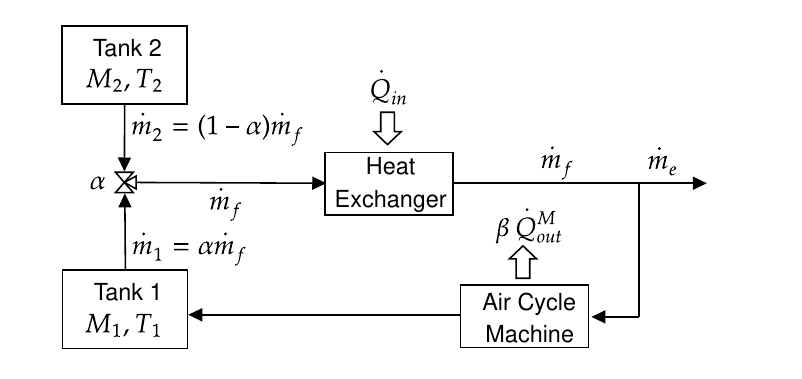}
    \caption{Dual-tank FTMS architecture modified with Air Cycle Machine.\hspace{9cm}}
    \label{fig:plantDiagram}
\end{figure}


Simulations have been conducted using an aircraft Fuel Thermal Management System (FTMS) model, depicted in Fig.~\ref{fig:plantDiagram}. This model is derived from the simplified model used by Leister and Koeln\cite{leisterNonlinearHierarchicalMPC2020} with the ram air cooler replaced by an Air Cycle Machine (ACM) to provide additional cooling power. The model can be derived from first principles, resulting in the continuous-time dynamics
\begin{subequations}\label{eq:FTMSdynamics}
\begin{align}
    \dot{M}_1 &= (1-\alpha)\dot{m}_f - \dot{m}_e, \label{eq:M1dot}\\
    \dot{M}_2 &= -(1-\alpha)\dot{m}_f,  \label{eq:M2dot} \\
    \dot{T}_1 & = \frac{(\dot{m}_f - \dot{m}_e)}{M_1} \left[ (1-\alpha)(T_2 - T_1) + \frac{	\dot{Q}_{in} }{c_v\dot{m}_f}\right] - \beta\frac{\dot{Q}_{out}^{M}}{c_vM_1}.\label{eq:T1dot}
\end{align}
\end{subequations}
The model states $x = \begin{bmatrix} M_1 & M_2 & T_1 \end{bmatrix}^\top$ correspond to the fuel mass in Tanks 1 and 2 and the fuel temperature in Tank 1, respectively. The control inputs $u= \begin{bmatrix} \alpha & \beta \end{bmatrix}^\top$ are the split-valve opening and the relative ACM load, respectively. Note that $\dot{Q}_{out}^{M}$ is the maximum ACM load such that the term $\beta\dot{Q}_{out}^{M}$ represents the amount of cooling provided by the ACM at any point in time. The only disturbance in the model is $\dot{Q}_{h_v}$, which is the only time-varying component of the total FTMS heat load $\dot{Q}_{in}$, which in turn is given by
\begin{equation}\label{eq:Qin}
    \dot{Q}_{in} = \dot{Q}_F + \dot{Q}_{h_v} + \dot{Q}_{h_e} + \left( P_p + K_{Q_h} \dot{m}_f \right).
\end{equation}
All the other terms in \eqref{eq:Qin} are considered constant. The reader is referred to the work by Leister and Koeln\cite{leisterNonlinearHierarchicalMPC2020} for more details on the physical interpretation of the remaining parameters in this model. Table~\ref{tab:ModelParams} summarizes the values used for each model parameter.

\begin{table}[htb]
\scriptsize
\caption{\uppercase{FTMS model parameters.}\hspace{15cm}}
\begin{center}
\label{tab:ModelParams}
\scalebox{1.00}{
\begin{tabular}{p{0.6cm}<{\centering} p{3cm}<{\centering} p{0.7cm}<{\centering} p{0.7cm}<{\centering} p{0.4cm}<{\centering} p{0.4cm}<{\centering}}
\hline
Variable & Description & Units & Nominal Value* & Lower Bound & Upper Bound \\ \hline
$M_1$ & Recirculation tank mass              &  kg   & 200   & 50        & 2850 \\
$M_2$ & Reservoir tank mass                  & kg    & 2850  & 50         & 2850 \\
$T_1$ & Recirculation fuel temperature       &  K    & 288   & 250       & 333 \\
$T_2$ & Reservoir fuel temperature            &  K    & 288   &-       &-\\
$\dot{m}_f$ & Pumped fuel flow rate               &  kg/s   & 1.0    &-          & -     \\
$\dot{m}_e$ & Engine fuel flow rate               & kg/s    & 0.26  & -         & -     \\
$\alpha$ & Recirculation fuel fraction            & -     & -     & 0         & 1.0   \\
$\beta$ & Relative ACM load                     & -     & -     & 0         & 1.0   \\
$c_v$ & Fuel specific heat                        & J/(kg$\cdot$K)&2,010      &-  &-\\
$\dot{Q}_F$ & FADEC heat input                    &W     &1,000    &-          &-      \\
$\dot{Q}_{h_v}$ & VCS heat input             &W     &55,000    &-          &-      \\
$\dot{Q}_{h_e}$ & Engine heat input               &W     &10,000    &-          &-      \\
$\dot{Q}_{out}^{M}$ & Maximum ACM load              &W     &120,000    &-          &-      \\
$P_p$ & Fuel pump power                            &W     &50,000    &-          &-      \\
$K_{Q_h}$ & Fuel pump heat input coeff.            &W/kg  &-6,618    &-          &-      \\
\multicolumn{6}{c}{*Value also serves as the initial state for $ M_1 $, $ M_2 $, and $ T_1 $.}
\end{tabular}}
\end{center}
\end{table}


The discrete-time nonlinear model used in the simulations from this section is obtained by simple forward Euler discretization of the continuous-time dynamics \eqref{eq:FTMSdynamics}.

\subsection{Controller implementations} \label{sec:ControllerImplementations}
Since the controller proposed in \AlgGenericRNMPC{} requires the iterative solution of a NLP, that is \RNMPCProblem{}, the use of different solvers results in different controller implementations. All simulations were performed in Matlab using CasADi \cite{anderssonCasADiSoftwareFramework2019} to formulate the optimization problems. In this work, three different implementations are compared. The first two use IPOPT \cite{byrdInteriorPointAlgorithm1999} to solve \RNMPCProblem{} and are named \IPOPTOne{} and \IPOPTTwo{}. The difference between both lies in how equation \eqref{eq:T1dot} is entered in the optimization problem formulation. Note that, once the Euler discretization is applied, the equality constraints related to \eqref{eq:T1dot} have bilinear terms on the right-hand side divided by the decision variables related to $M_1$. This original formulation is used by \IPOPTTwo{}. In \IPOPTOne{} both sides of these equality constraints are multiplied by the corresponding discretized decision variables related to $M_1$. This results in equality constraints that have only bilinear terms on both sides of the equality. As will be seen in the remainder of this section, even though both approaches are mathematically equivalent, the reformulation in \IPOPTOne{} improves the performance of IPOPT. The third implementation of the controller, \SLRNMPC{}, uses a custom NLP solver based on Successive Linearization (SL), very similar to the SL implementation described by Leister and Koeln\cite{leisterNonlinearHierarchicalMPC2022} and originally presented by Mao et al\cite{maoSuccessiveConvexificationNonconvex2016}. The reader is referred to the work of Leister and Koeln\cite{leisterNonlinearHierarchicalMPC2022} and references therein for more details on the SL algorithm. As with \IPOPTTwo{}, the \SLRNMPC{} uses the original formulation of the dynamics equations. The SL algorithm internally converts \RNMPCProblem{} to a series of Quadratic Programs (QPs) and Gurobi \cite{gurobioptimizationllcGurobiOptimizerReference2021} is used to solve these QPs. Computations were performed on a desktop computer with a \mbox{3.2 GHz} \mbox{i7-8700} processor and \mbox{16 GB} of RAM.

The initial valid reference trajectory for the controllers is computed offline using the strategy suggested in Remark \ref{rem:InitialRefTrajectory}. Also, the following parameters are used to compute the LQR controller matrices: $ Q = diag\left(\begin{bmatrix}  1/500 & 1/100 & 40/300   \end{bmatrix}\right)$  and $ R = diag\left(\begin{bmatrix} 1 & 0.01 \end{bmatrix}\right)$. The parameter \texttt{maxIters} in \AlgOptimValidRefTrajectory{} is set to 20 for all controllers. The output vectors here are defined to be simply $y_k = \rvect{x_k^\top \; u_k^\top}^\top$ and matrices $C_k$ and $D_k$ are defined accordingly. The output constraints $\Yc_k$ are defined as constant intervals matching, for all time steps, the state and input constraints defined in Table \ref{tab:ModelParams}.

The objective function used in \RNMPCProblem{} is
\begin{align*} 
    \ell(\Xseq_k,\Useq_k) = \sum_{i=0}^{N-k-1}  \Delta u_{k+i|k}^\top\Delta u_{k+i|k} + w_{\beta} \beta_{k+i|k}^2,
\end{align*}
where  $w_{\beta}=5$ is a weighting factor and
\begin{align*} 
    \Delta u_{k+i|k} =  (u_{k+i|k} - u_{k+i-1|k}).
\end{align*}
For $i=0$, the term $u_{k+i-1|k}$ refers to the input vector at the previous time step. This objective function is designed mainly to incentivize minimum energy use by penalizing the use of the ACM input $\beta$ and also to encourage smoothness in the control inputs by penalizing the difference between the input vectors at subsequent time steps.  

Finally, the controllers implemented here are tested with two different sampling times: $T_s=50$ seconds and $T_s=100$ seconds. Since all simulations are performed with a fixed final time of 10,000 seconds, the underlying optimization problems have about twice as many decision variables when $T_s=50$ seconds. This allows some evaluation of how well these controllers scale with the optimization problem size. However, to allow fair comparisons of the performance of controllers with different sampling times, the following equalized objective function is used
\begin{align*} 
    \ell_{eq}(\Xseq,\Useq) = T_s \cdot \ell(\Xseq,\Useq).
\end{align*}
This equalized objective function is not implemented in the controllers but is only used for comparisons in Section \ref{sec:PerformanceComparison}. These sampling times are also used to obtain the discrete-time models from the forward Euler discretization of the continuous-time model.

\subsection{Test cases}
Recall that the controller formulation in Section \ref{sec:ProposedController} considers that the actual disturbance $d_k$ hitting the system has two components: a reference disturbance $\dref_k$, which can be seen as the predicted disturbance, and an unknown disturbance deviation $\Dd_k$ bounded by known sets $\Dc_k$. All test cases simulated in this work consider a fixed interval $\Dc_k= [-27500,+27500]$ W, which is 50\% variation over the nominal disturbance variable value, $\dot{Q}_{h_v}$. Also, all test cases consider the same reference trajectory $\Dref_0$ shown in Fig. \ref{fig:Disturbances}. The test cases consider two specific realizations of $\Dd_k$: $\Dd_k$ as a square wave with amplitude given by the extremes of $\Dc_k$ and $\Dd_k$ as a random signal with uniform distribution and maximum amplitude given by $\Dc_k$. 


\begin{figure}[ht]
    \centering
    \includegraphics[width=0.5\linewidth] {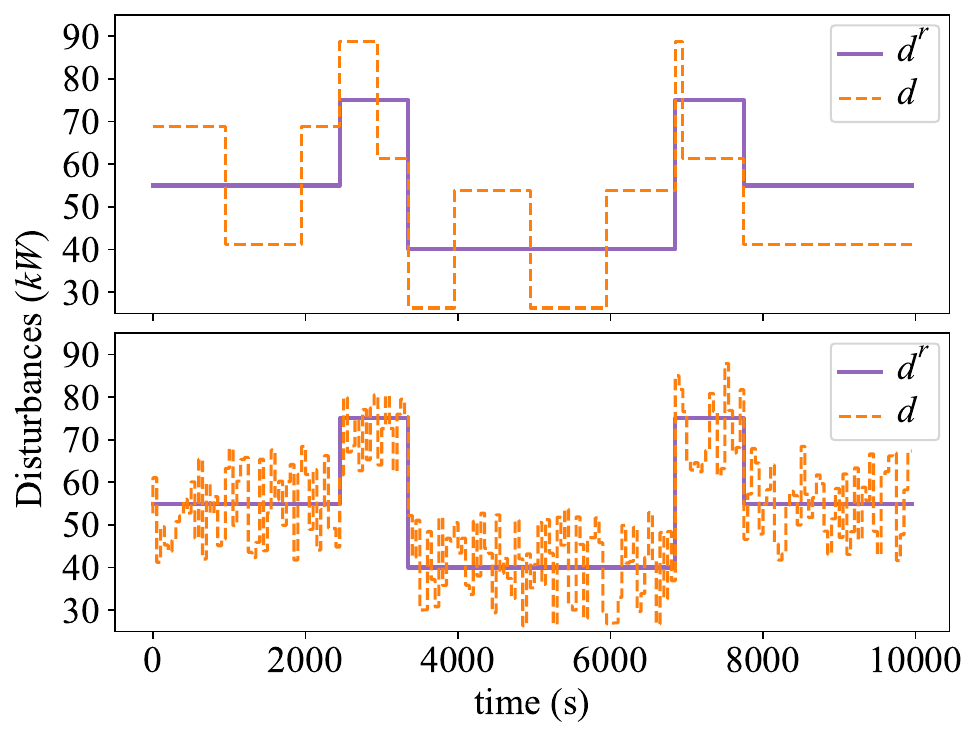}
    \caption{Disturbance profiles used in the simulations. Top: $\Dd_k$ is a square wave. Bottom: $\Dd_k$ is a random signal.\hspace{9cm}}
    \label{fig:Disturbances}
\end{figure}


The combination of these disturbance scenarios with the different sampling times mentioned in Section \ref{sec:ControllerImplementations} results in all test cases evaluated in this work. The specific case when the disturbance from the top of Fig. \ref{fig:Disturbances} is applied and $T_s=100$ seconds will be further referred to as \MainTestCase{} in order to highlight more detailed results in the remainder of this section.

To illustrate the need for a robust controller in these scenarios, Fig. \ref{fig:MainTestCase_NMPCFail} shows the resulting trajectory when a nominal (non-robust) NMPC controller is applied in closed-loop for \MainTestCase{}. This controller simply solves \RNMPCProblem{} once at each time step using IPOPT. However, the tightened constraints $\Yhc_k$  in \eqref{eq:OutputConstraint_MPC} are replaced by the original constraints $\Yc_k$, making it effectively a nominal NMPC. As seen in the figure, temperature constraint violations occur as early as $t=400$ seconds.


\begin{figure}[ht]
    \centering
    \includegraphics[width=0.5\linewidth] {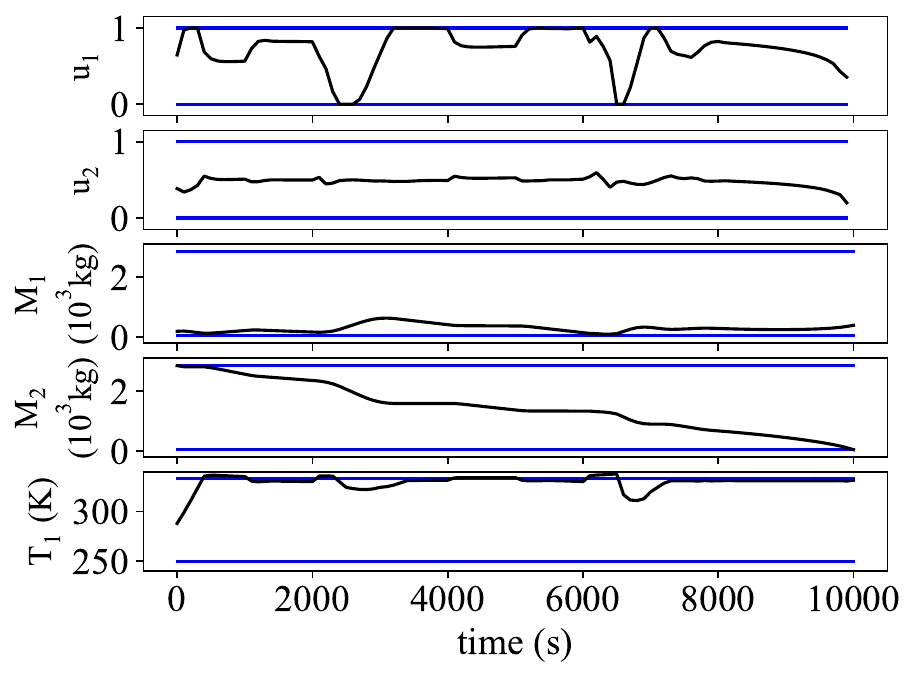}
    \caption{State and input trajectories obtained when a nominal NMPC is applied to \MainTestCase{}, with multiple temperature constraint violations. Original constraints are shown in blue.}
    \label{fig:MainTestCase_NMPCFail}
\end{figure}


\subsection{Error sets and fallback control}
Fig. \ref{fig:MainTestCase_ErrorSetsAtt0} shows an example of the error sets $\bigE_{k+i|k}$ computed for \MainTestCase{} at $k=0$, i.e., based on the initial reference trajectory. Since, as detailed in Section \ref{sec:ConsiderationsPracticalImplementation}, these error sets are approximated by intervals, they are represented in Fig. \ref{fig:MainTestCase_ErrorSetsAtt0} by their extreme values in each dimension, shown in red. Fig.~\ref{fig:MainTestCase_ErrorSetsAtt0} also shows the actual error $e_k$ obtained for 100 random realizations of $\Dd_k$ and with the FTMS system controlled using only the fallback control input \eqref{eq:MPCplusLQRControl} based on the same initial reference trajectory used to compute these error sets. This example highlights how the LQR control component in \eqref{eq:MPCplusLQRControl} can effectively limit the difference between the reference and actual states as well as how the computed error sets effectively bound these possible error trajectories. Note that two of the error trajectories (in black) actually align with the error set boundaries: these trajectories refer to two additional realizations of $\Dd_k$, namely when the disturbance is constantly at the upper or lower extremes of $\Dc_k$ respectively. This shows that, specifically for this application example, despite the overapproximations in the computation of the error sets, the sets obtained are not conservative, since there are disturbance realizations that bring the error trajectory very close to the error set boundaries.


\begin{figure}[ht]
    \centering
    \includegraphics[width=0.5\columnwidth] {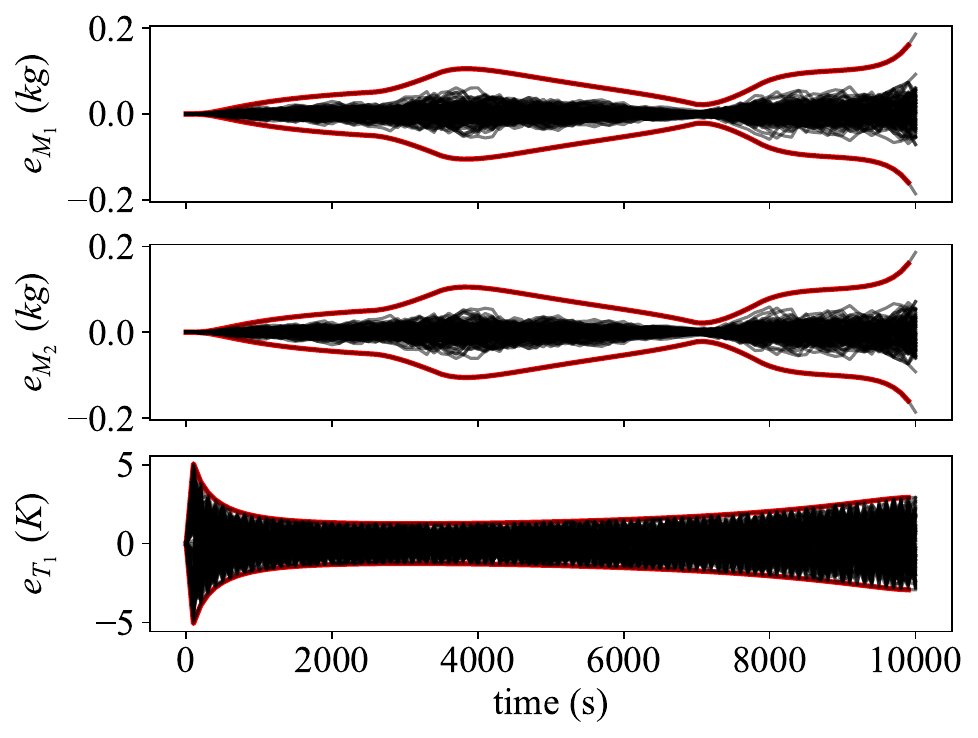}
    \caption{Error sets $\bigE_{k+i|k}$ computed for $k=0$ (red) and the actual error $e_k$ (black) obtained for 100 random realizations of $\Dd_k$ and with the FTMS system controlled using only the fallback control input \eqref{eq:MPCplusLQRControl} based on the same reference trajectory used to compute these error sets.}
    \label{fig:MainTestCase_ErrorSetsAtt0}
\end{figure}


\subsubsection{Effect of the fallback control on closed-loop trajectories and performance}
Fig. \ref{fig:SwitchingControl_MainTestCase_XUtrajectories} shows, for \MainTestCase{} and using \IPOPTOne{}, the effect of the use of the fallback control through three different situations: i) when no fallback control is used, ii) when the fallback control is alternately applied and not applied for periods of 500 seconds, and iii) when only the fallback control is used. It can be seen that, when only the fallback control is used, substantially different trajectories are obtained compared to the case with no fallback control. This highlights the fact that the \IPOPTOne{} is effectively optimizing the planned trajectory compared to the initial feasible trajectory. In fact, the normalized and equalized objective function values obtained are 1.002 with no fallback control, 1.008 in the alternating case, and 1.520 with fallback control only. Similarly to Fig. \ref{fig:AllTestCases_ObjFunctionValues}, the objective function value for \IPOPTTwo{} in Test Case 1 was used for normalization.

The input trajectories for the case when the fallback option is alternately applied highlight one potential effect of using the fallback control: the transitions when going from the fallback to the optimized control option may be considerably abrupt since the newly obtained optimized solution may be quite different from the current fallback option. These transitions can result in more oscillatory behavior in the state trajectories, although in the present example these are considerably low due to the high damping on the system dynamics. It should be noted that the transitions from the optimized to the fallback control option are in general not abrupt, since at these time steps the fallback control will be based on an optimized trajectory that was obtained just at the previous time step.

\begin{figure}[ht] 
        \centering
        \includegraphics[width=0.5\columnwidth]{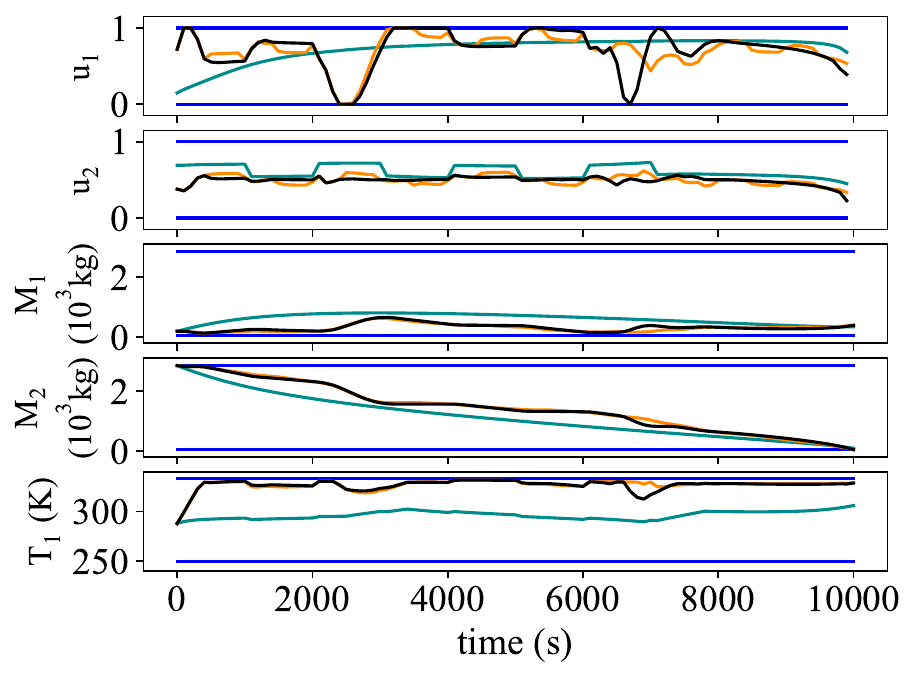}
    \caption{State and input trajectories obtained when applying \IPOPTOne{} (black), \IPOPTOne{} with the fallback control applied alternately for periods of 500 seconds (orange), and fallback control only (dark green) to \MainTestCase{}. Original constraints are shown in blue.}
\label{fig:SwitchingControl_MainTestCase_XUtrajectories}
\end{figure}

\subsection{Closed-loop trajectories}
Fig. \ref{fig:MainTestCase_AllControllers_XUtrajectories} shows the closed-loop state and input trajectories obtained when using each of the three controller implementations for \MainTestCase{}. \IPOPTOne{} and \IPOPTTwo{} have very similar trajectories since they are essentially using the same solver. The \SLRNMPC{} however presents some clear differences in the trajectory obtained, which is explained by the fact that a different solver is being used and each solver may eventually converge to different local minima at the very first time step when \RNMPCProblem{} is solved in \AlgGenericRNMPC{}, leading to different closed-loop trajectories. It is also clear from Fig. \ref{fig:MainTestCase_AllControllers_XUtrajectories} that all three controllers are capable of ensuring robust constraint satisfaction throughout the entire simulation.

To illustrate the amount of spread on the closed-loop trajectories when different realizations of the disturbance are applied, Fig. \ref{fig:RandomTestCase_XUtrajectories} shows the corresponding trajectories obtained for 100 different random realizations of $\Dd_k$. The trajectories for \IPOPTTwo{} are omitted for conciseness. Interestingly, \IPOPTOne{} seems to be less consistent in the solutions obtained. Specifically, when observing the behavior of the $\alpha$ input, there is more variation in the input profile compared to the $\alpha$ inputs from \SLRNMPC{}. As expected though, robustness is achieved for all simulated realizations of the disturbance for both controllers.

\begin{figure}[ht] 
        \centering
        \includegraphics[width=\textwidth]{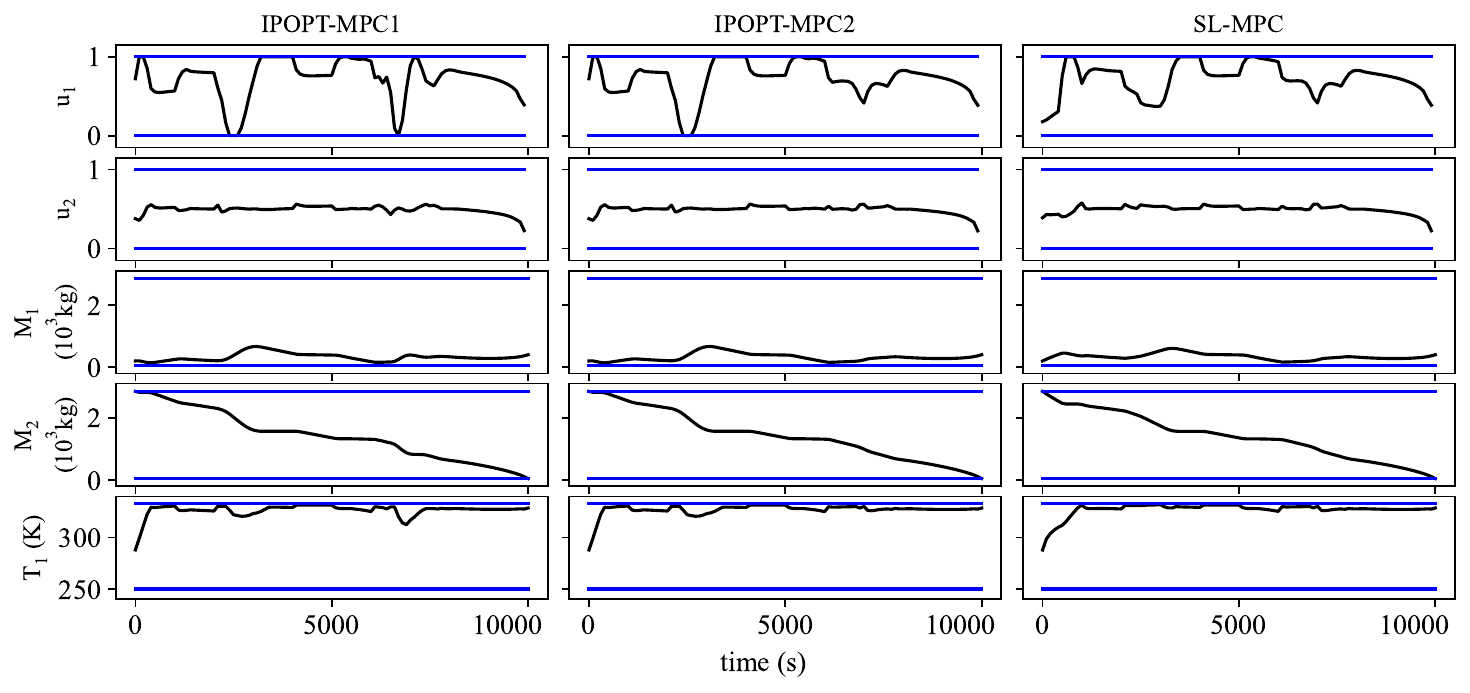}
    \caption{State and input trajectories obtained when applying the three robust NMPC controllers to \MainTestCase{}. Original constraints are shown in blue.}
\label{fig:MainTestCase_AllControllers_XUtrajectories}
\end{figure}
\begin{figure}[ht] 
    \begin{subfigure}[b]{0.5\linewidth}
        \centering
        \includegraphics[width=\textwidth]{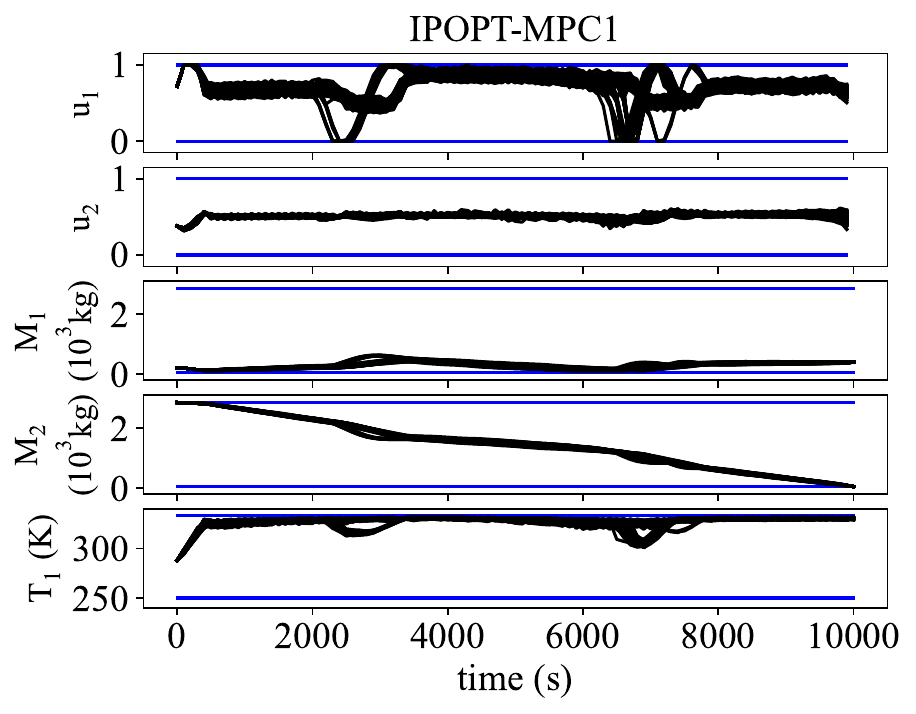}
        \label{fig:RandomTestCase_trajectories_IPOPT1}
    \end{subfigure}%
   ~
    \begin{subfigure}[b]{0.5\linewidth}
        \centering
        \includegraphics[width=\textwidth]{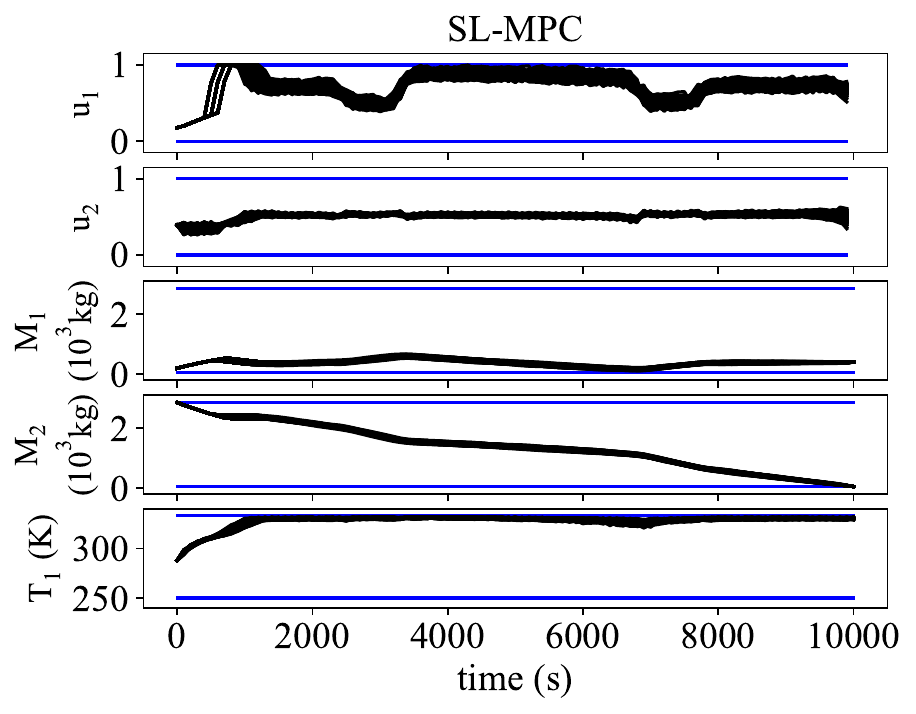} 
        \label{fig:RandomTestCase_trajectories_SLRNMPC}
    \end{subfigure}%
    \caption{State and input trajectories obtained when applying \IPOPTOne{} and \SLRNMPC{} in closed-loop to a 100 different random realizations of $\Dd_k$ and with $T_s=100$. Original constraints are shown in blue.}
\label{fig:RandomTestCase_XUtrajectories}
\end{figure}

\subsection{Iterations in \AlgOptimValidRefTrajectory{}} \label{sec:IterationsInAlgRefTraj}

Since the iterations in \AlgOptimValidRefTrajectory{} are not guaranteed to converge to a new optimized valid reference trajectory, Fig.~\ref{fig:MainTestCase_NumIters} gives some insight into the behavior of the three controller implementations in this regard. The figure shows, for each time step in the closed-loop simulations for \MainTestCase{}, how many iterations each controller took to converge to a solution, if converging at all. While \IPOPTTwo{} was able to find valid solutions within the allowed maximum number of iterations \IPOPTOne{} failed to converge within the maximum number of iterations in two subsequent time steps. The \SLRNMPC{} converged to an invalid solution at a few time steps at the beginning of the simulation within just one iteration. This issue stems from the nature of the SL algorithm, where it is possible that the algorithm converges quickly to an infeasible solution. Overall, there is a tendency of the \SLRNMPC{} to require slightly fewer iterations to converge compared to the IPOPT-based controllers. This likely helps the \SLRNMPC{} use less computation time, as detailed in Section \ref{sec:PerformanceComparison}. 


\begin{figure}[ht]
    \centering
    \includegraphics[width=0.9\columnwidth] {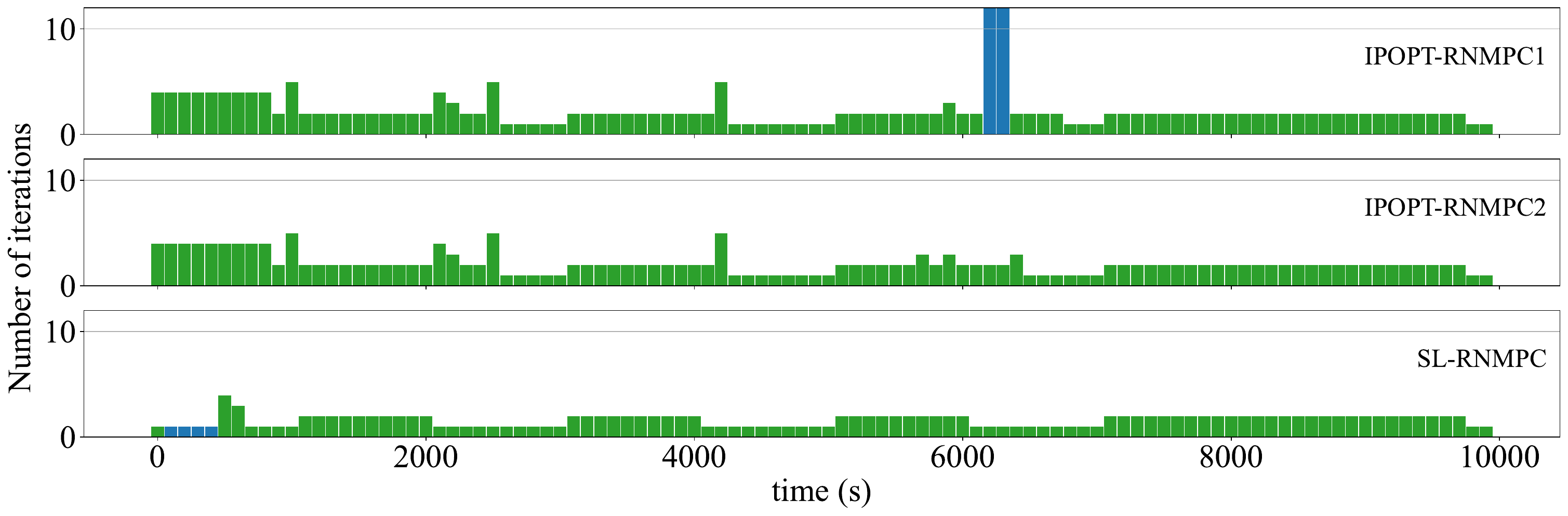}
    \caption{Number of iterations taken in \AlgOptimValidRefTrajectory{} by each controller for \MainTestCase{}. Time steps when the controllers failed to converge to a valid solution are shown in dark blue.}\hspace{9cm}
    \label{fig:MainTestCase_NumIters}
\end{figure}


\vspace{25pt}
\subsection{Controller performance }\label{sec:PerformanceComparison}
In this section, the performance of the proposed controllers is compared across multiple scenarios as well as with the work of Leeman et al\cite{leemanRobustOptimalControl2023}. 
\subsubsection{Comparison among implementations of the proposed controller}
Figs. \ref{fig:AllTestCases_ObjFunctionValues} to \ref{fig:AllTestCases_CompTimesMax} provide a comparison of the performance of all three controllers in terms of equalized objective function values, computed based on the resulting closed-loop trajectories, and computation times. From Fig. \ref{fig:AllTestCases_ObjFunctionValues}, it can be seen that the three controllers were able to achieve very similar objective function values. 

The similar objective function performance though comes with disparities in computational costs where, in terms of average computation time, the \SLRNMPC{} is 15\% to 56\% faster than \IPOPTOne{}. In terms of the maximum computation time, the results between both controllers are mixed, with a slightly higher value for the \SLRNMPC{} at $T_s=50$ and lower values otherwise. Another difference between the controllers, not shown in the figures, is the fraction of time that the controller spends solving the optimization problem as opposed to computing the constraint tightening for each new iteration: \IPOPTOne{} spent from 33\% to 50\% of the time in the solver while \SLRNMPC{} spent 16\% to 26\%.

\IPOPTTwo{} has clearly much higher computation times, both average and maximum, compared to the other two controllers. This stems from the fact that \IPOPTTwo{} is using the original equality constraint formulation in the optimization problem formulation, as discussed in Section \ref{sec:ControllerImplementations}. This formulation not only results in about twice the average computation times compared to \IPOPTOne{} but also much higher discrepancies in the maximum computation times. These differences between \IPOPTOne{} and \IPOPTTwo{} highlight one advantage of using \SLRNMPC{}: the optimization problem formulations in \SLRNMPC{} required no further manipulation of the original nonlinear dynamics function to obtain lower average computation times, as opposed to the manipulations done to implement \IPOPTOne{}.


\begin{figure}[ht]
    \centering
    \includegraphics[width=0.7\columnwidth] {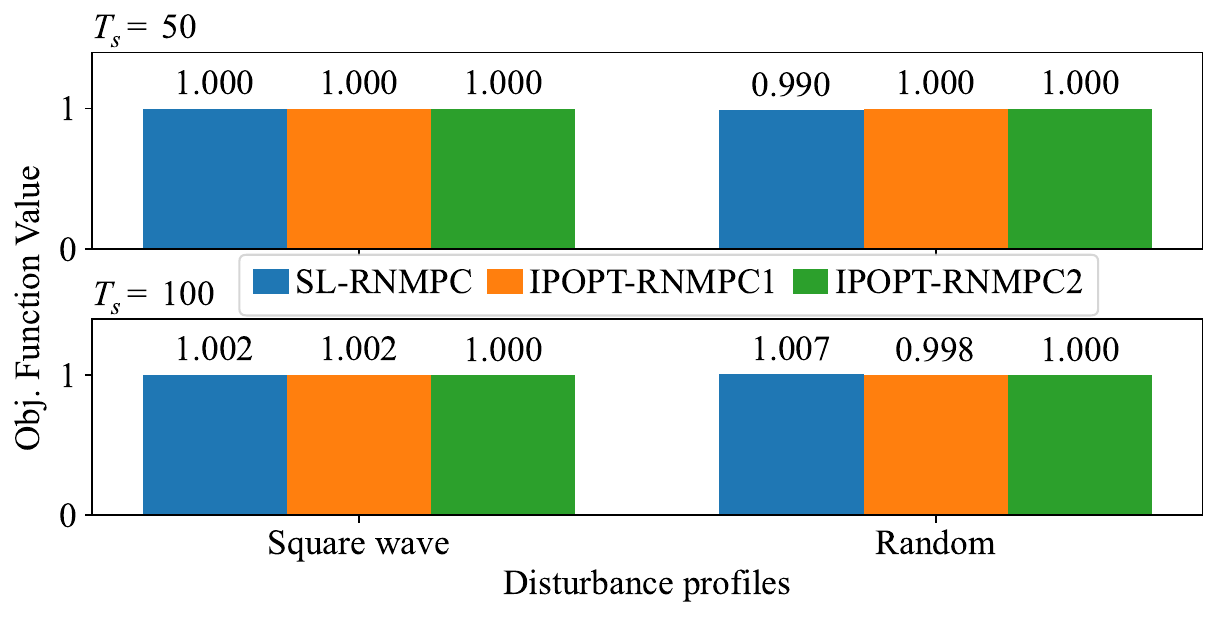}
    \caption{Equalized and normalized objective function values obtained for all test cases for the three controllers. The values shown here are normalized by the corresponding values obtained for the \IPOPTTwo{} in each test case. }
    \label{fig:AllTestCases_ObjFunctionValues}
\end{figure}



\begin{figure}[ht]
    \centering
    \includegraphics[width=0.7\columnwidth] {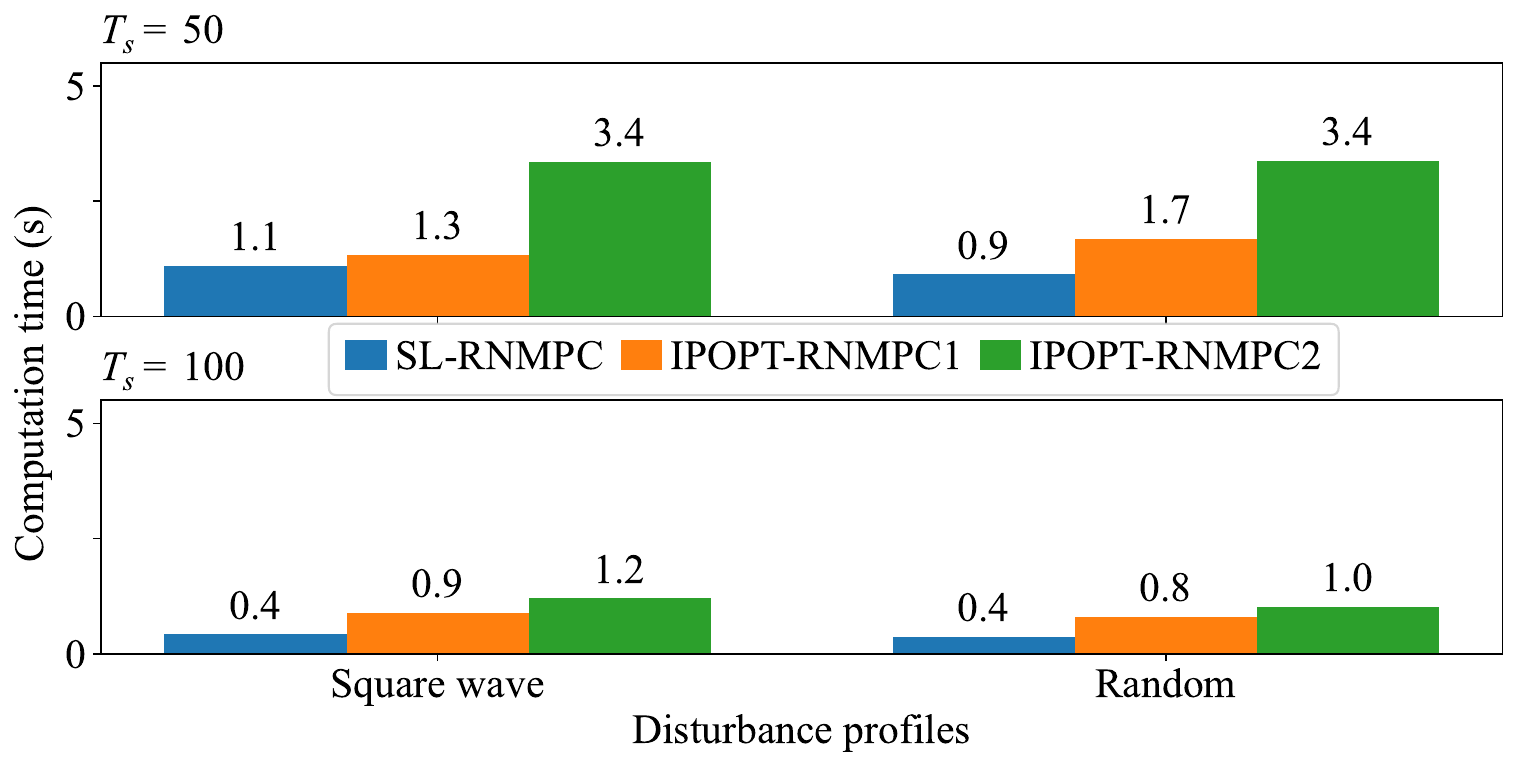}
    \caption{Average computation times obtained for all test cases for the three controllers.\hspace{7cm} }
    \label{fig:AllTestCases_CompTimesAvg}
\end{figure}



\begin{figure}[ht]
    \centering
    \includegraphics[width=0.7\columnwidth] {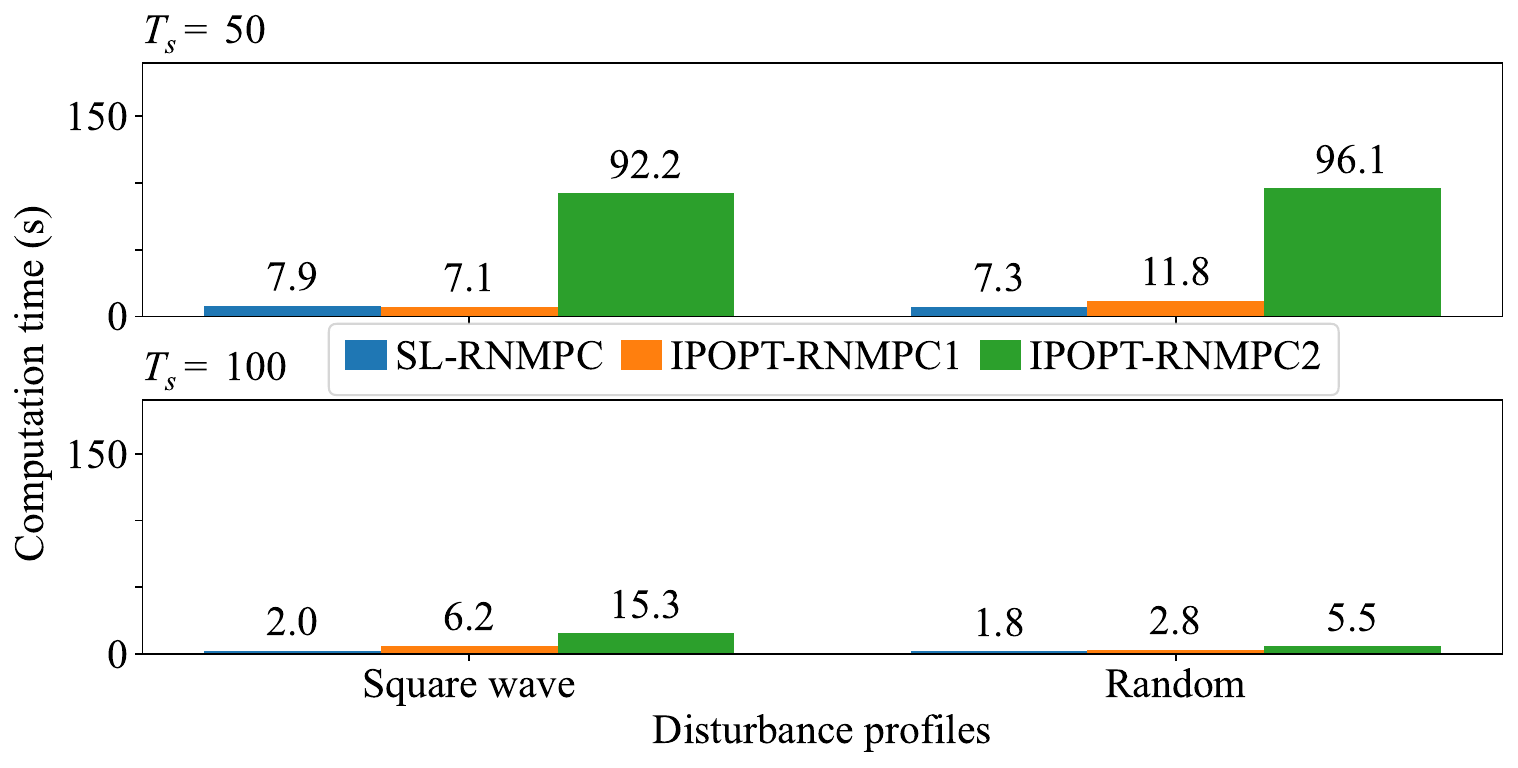}
    \caption{Maximum computation times obtained for all test cases for the three controllers.\hspace{9cm}}
    \label{fig:AllTestCases_CompTimesMax}
\end{figure}


\subsubsection{Comparison with existing technique \cite{leemanRobustOptimalControl2023}}

This section provides a comparison of the proposed formulation, specifically the \IPOPTOne{} implementation, with the controller formulation due to Leeman et al\cite{leemanRobustOptimalControl2023}, herein referred to as \SLS{} controller. Similarly to the formulation proposed in this work, \SLS{} controller relies on an LTV approximation of the nonlinear system dynamics. The effect of linearization errors, external disturbances, and parametric uncertainties is encoded into the parameters of the optimization problem and used to tighten the original constraints. Therefore, in their formulation, based on the concept of System Level Synthesis (SLS), both the nominal trajectory as well as the state feedback control law are optimized jointly. This differs from the strategy adopted in this work, where the optimized nominal trajectory and feedback control gains $K_{k+i|k}$ are obtained iteratively in separate steps of Algorithm \ref{alg:OptimValidRefTrajectory}.

In both techniques a key aspect is the size of the reachable sets of the error dynamics, i.e., the LTV system that represents deviations from the nominal optimized trajectory and when controlled with the LTV state-feedback controller. In this work these sets correspond to the $\bigE_{k+i|k}$ sets for the states and $K_{k+i|k}\bigE_{k+i|k}$ sets for the inputs. These sets are used for constraint tightening and therefore, the larger the sets the more conservatism is present in the overall solution. The comparison provided here focuses on the solution to the robust optimal control problem for a single time step, particularly in terms of the propagation of the error sets and the corresponding computation times. As a benchmark application example, the satellite post-capture stabilization example provided in Leeman et al\cite{leemanRobustOptimalControl2023} is replicated here. The reader is referred to that reference for more details on this application example. The results for the \SLS{} were obtained through the use of the MATLAB code provided by the authors\footnote{\url{https://gitlab.ethz.ch/ics/nonlinear-parametric-SLS}}. The same discrete-time model was used to obtain results for \IPOPTOne{} with the following exception: Leeman et al\cite{leemanRobustOptimalControl2023} explicitly deal with parametric uncertainties whereas in the implementation for \IPOPTOne{} such parametric uncertainty was converted into an additional bounded disturbance. The computations with \IPOPTOne{} were set to provide as much of an accurate comparison with the \SLS{} controller as possible: all model parameters used are the same as described in Leeman et al\cite{leemanRobustOptimalControl2023}. The same holds for the initial guess and the objective function used. The $Q$, $Q_f$, and $R$ matrices of the objective function are also used as the weighting matrices for the LQR controller in \IPOPTOne{}.  

Fig. \ref{fig:PlanarRigidBody_SetLimits} shows the resulting state and input reachable sets of the error dynamics considering a prediction horizon of 5 seconds ($N=10$) and 10 seconds ($N=20$). It can be seen, in general, that the reachable sets obtained using \IPOPTOne{} are smaller, or of about the same size, compared to the sets from the \SLS{} controller. Two exceptions are the input sets for $T=10$ seconds during the first half of the prediction horizon, where \IPOPTOne{} is more conservative. Notice that the sets obtained for the first half of the case when $T=10$ seconds do not need to match the sets for $T=5$ seconds since different time-varying feedback gains are used for each controller. Also, there are no sets shown for \IPOPTOne{} for the inputs $v_x$ and $v_y$ at the last time steps because the formulation in this work produces inputs up to time step $N$ while in Leeman et al\cite{leemanRobustOptimalControl2023} inputs up to time step $N+1$ are produced. Table \ref{tab:RigidBodyComparisonTable} shows the computation times for each case. Not only has \IPOPTOne{} a much lower computation time, about 50 times less for $T=5$ seconds, but also seems to scale better: doubling the horizon length essentially doubled the computation time for \IPOPTOne{} but caused a more than 20-fold increase for the \SLS{} controller. The objective function values of each nominal trajectory obtained are also shown in Table \ref{tab:RigidBodyComparisonTable} to highlight the fact that the use of \IPOPTOne{} did not cause degraded performance in terms of objective function values. For clarity, the nominal optimized trajectories obtained are not depicted in Fig. \ref{fig:PlanarRigidBody_SetLimits}.

Overall, these results suggest that the proposed formulation has the potential to generate less conservative trajectories, requires much lower computation times, and scales better with the prediction horizon all while providing guarantees of robust constraint satisfaction. These advantages come at the cost of not having guaranteed convergence to an optimal solution, which may require the application of the suboptimal fallback control option. Moreover, the \SLS{} controller has the advantage of providing robust performance guarantees\cite{leemanRobustOptimalControl2023}. The results obtained here show that the proposed formulation has promising potential, particularly in applications that have tight computation time requirements or long prediction horizons.

\begin{figure}[ht] 
    \begin{subfigure}[b]{0.5\linewidth}
        \centering
        \includegraphics[width=\textwidth]{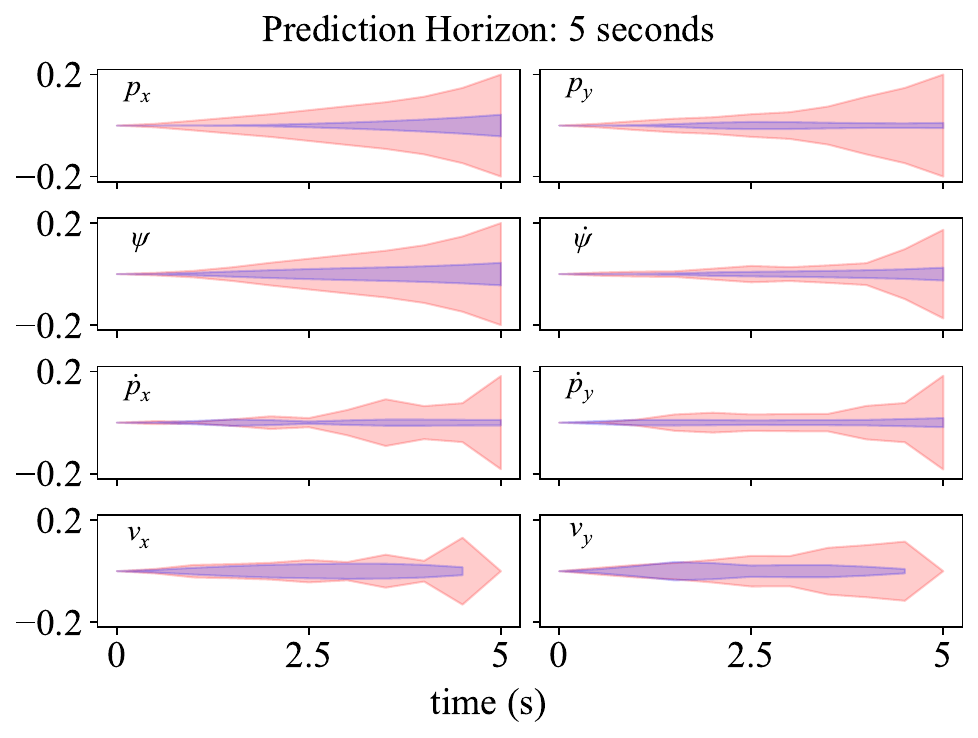}
        \label{fig:PlanarRigidBody_tfinal5_SetLimits}
    \end{subfigure}%
   ~
    \begin{subfigure}[b]{0.5\linewidth}
        \centering
        \includegraphics[width=\textwidth]{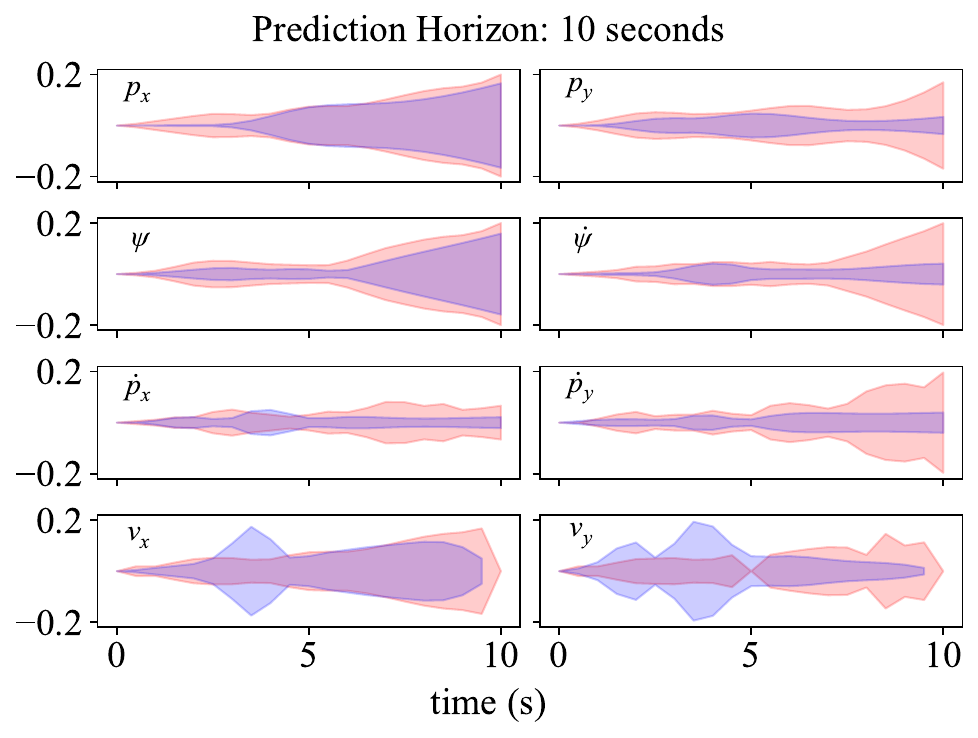} 
        \label{fig:PlanarRigidBody_tfinal10_SetLimits}
    \end{subfigure}%
    \caption{Propagation of the reachable sets of the error dynamics for the \IPOPTOne{} (blue) and \SLS{} (red) controllers for the satellite post-capture stabilization problem\cite{leemanRobustOptimalControl2023} with horizons of 5 and 10 seconds. The reader is referred to Leeman et al\cite{leemanRobustOptimalControl2023} for the meaning of each of the state and input variables shown here.}
\label{fig:PlanarRigidBody_SetLimits}
\end{figure}




\begin{table}[htb]
\scriptsize
\caption{\uppercase{Performance of the \IPOPTOne{} and the \SLS{} controller.\hspace{15cm}}}
\begin{center}
\label{tab:RigidBodyComparisonTable}
\scalebox{1.00}{
\begin{tabular}{p{2cm}<{\centering} p{2cm}<{\centering} p{2cm}<{\centering} p{2cm}<{\centering} p{2cm}<{\centering} }
\hline
\multirow{2}{*}{Controller}  & \multicolumn{2}{c}{$T=5$s} & \multicolumn{2}{c}{$T=10$s}  \\ 
            & Comp. Time (s)    & Obj. Func. Value  &  Comp. Time (s)    & Obj. Func. Value \\ \hline
\IPOPTOne{} &  \textbf{1.03} & 17.29  &  \textbf{1.96 }  &  17.69\\       
\SLS{}      &  49.19 & 21.63  &   1192.4  &  21.96\\ \hline 
\end{tabular}}
\end{center}
\end{table}



 \section{Conclusion} \label{sec:conclusion}
A novel shrinking-horizon robust MPC formulation for nonlinear discrete-time systems was presented. The proposed controller iteratively solves a NLP with tightened constraints to obtain reference trajectories that are used to provide optimized operation with guaranteed robust state and input constraint satisfaction guarantees for all future time steps. When the iterations fail to produce a valid reference trajectory, a suboptimal fallback control option is used that preserves the guaranteed constraint satisfaction.
The proposed controller was tested, with three different NLP solvers, using an aircraft FTMS model under different disturbance scenarios. All three controller implementations provided robust optimized system operation with overall less computational load from the controller based on a Successive Linearization solver when compared to the controllers that use IPOPT to solve the underlying NLP. A comparison with one of the existing techniques in the literature showed promising results in terms of conservatism, computation time, and scalability. Future work will focus on extending the proposed robust controller formulation to the receding-horizon case.


\bmsection*{Acknowledgments}

This material is based upon work supported by the Office of Naval Research under award number N00014-22-1-2247. Any opinions, findings, and conclusions or recommendations expressed in this material are those of the authors and do not necessarily reflect the views of the Office of Naval Research.



\DeclareRobustCommand{\VAN}[3]{#3}

\bibliography{root}


\appendix
\bmsection{LQR Controller} \label{sec:LQRControllerEqs}
The following equations can be used to derive a LTV LQR controller for a LTV system with dynamics matrices $A_k$ and $B_k$ using a dynamic programming approach. The $Q$ and $R$ matrices are tunable parameters. The $K_{k+i|k}$ gains are to be computed backwards in time as
\begin{align*}
    P_{k+N|k} &= Q, \label{eq:LQR_PN}\\
    P_{k+i|k} &= Q + A_{k+i}^T P_{k+i+1|k}A_{k+i} - A_{k+i}^T P_{k+i+1|k}B_{k+i}(R + B_{k+i}^T P_{k+i+1}B_{k+i})^{-1} B_{k+i}^TP_{k+i+1|k}A_{k+i}, \\
    K_{k+i|k} &= (R + B_{k+i}^T P_{k+i+1}B_{k+i})^{-1}B_{k+i}^TP_{k+i+1|k}A_{k+i}. 
\end{align*}

\end{document}